\documentclass[manuscript,screen]{acmart}

\AtBeginDocument{%
  \providecommand\BibTeX{{%
    \normalfont B\kern-0.5em{\scshape i\kern-0.25em b}\kern-0.8em\TeX}}}

\setcopyright{acmlicensed}

\usepackage{graphicx} 
\usepackage{amsmath,verbatim,mathrsfs,latexsym,paralist}
\usepackage{indentfirst}
\usepackage{color}
\usepackage{amsthm}
\usepackage{xspace}
\newtheorem{Thm}{Theorem}
\newtheorem{Prop}{Property}
\newtheorem{lem}{Lemma}

\newtheorem{Es}{Example}

\usepackage{algorithm}

\usepackage{algorithm}
\usepackage{algpseudocode}
\usepackage{mathtools}
\usepackage[colorinlistoftodos,prependcaption,textsize=tiny]{todonotes}
\usepackage{calrsfs}
\usepackage{dutchcal}
\usepackage[most]{tcolorbox}

\setcopyright{acmlicensed}
\copyrightyear{2018}
\acmYear{2018}
\acmDOI{XXXXXXX.XXXXXXX}
\acmISBN{978-1-4503-XXXX-X/2018/06}
\begin{document}

\title{Optimal and heuristic strategies for evaluating the influence of coordinated behavior in information cascades and retweet networks}

\author{Niccolò Di Marco}
\email{niccolo.dimarco@unitus.it}
\orcid{0000-0003-4335-7328}

\author{Matteo Cinelli}
\email{matteo.cinelli@uniroma1.it}
\orcid{0000-0003-3899-4592}

\author{Shinichi Nakano}
\email{nakano@gunma-u.ac.jp}
%\orcid{0000-0003-3899-4592}

\author{Andrea Frosini}
\email{andrea.frosini@unifi.it}
\orcid{0000-0001-7210-2231}

\renewcommand{\shortauthors}{Di Marco et al.}

\begin{abstract}
Coordinated Inauthentic Behavior (CIB) has become a major concern in online social platforms, yet its actual impact on information diffusion remains poorly understood. Existing research has primarily focused on detecting coordinated activity, while comparatively little attention has been devoted to quantifying its influence once detected.
In this work, we introduce two complementary frameworks for the post-hoc evaluation of coordinated accounts. First, we formulate the problem on information cascades as a constrained influence maximization problem over directed trees and develop a polynomial-time dynamic programming algorithm that computes the optimal placement of coordinated nodes, providing an upper bound on their achievable influence. Second, motivated by the limited availability of diffusion cascades in real-world platforms, we propose a network-based framework that estimates influence directly from retweet networks using the independent cascade model and compares the observed placement of coordinated accounts against established heuristic baselines.
We evaluate both approaches on Twitter/X data from the 2019 UK General Election and on a collection of verified state-backed information operation campaigns spanning multiple countries. While coordinated accounts exhibit limited influence in the UK cascades, the network-based analysis reveals substantial differences across campaigns, with several operations achieving influence comparable to or exceeding that of structurally central seed sets. Finally, by reconstructing cascades from the retweet networks, we show that the two frameworks produce consistent results, suggesting that the observed effects reflect intrinsic structural properties of coordinated activity rather than artifacts of the underlying methodology.
%
% Our work provides a unified framework for quantifying the influence of coordinated behavior across both diffusion cascades and retweet networks, offering practical tools for assessing the potential impact of coordinated information operations.
\end{abstract}





\maketitle

\section{Introduction}

The advent of social media has radically reshaped how people produce, share, and consume information. Crucially, it has also transformed how influence and propaganda are delivered to users \cite{ferrara2016rise, woolley2020bots}. Although the actual impact of such information operations is debated—with journalists and governments claiming definitive effects on political outcomes (e.g., \cite{guardian_russia_2024}) and scientific work often finding a lack of measurable effects \cite{woolley2020bots, eady2023exposure}—the underlying fact remains that these campaigns do take place \cite{Cresci2020decade,serafino2024suspended, CIB_Survey, hristakieva2022spread, Cinus2025exposing, Loru2026}. Consequently, within this heated debate, detecting potentially malicious campaigns in online environments is fundamental to preserving information and electoral integrity, and, more broadly, safeguarding democracy and freedom on the web.

Many of these campaigns manifest in the form of coordinated inauthentic behavior (CIB), defined by Meta as “the use of multiple assets, working in concert to engage in inauthentic behavior, where the use of fake accounts is central to the operation” \cite{gleicher2018coordinated}. More broadly, inauthentic behavior refers to practices such as misrepresenting identity, artificially amplifying content, or enabling other violations of platform policies. Closely related, and more easily measurable using data available from private social media companies \cite{wagner2023independence}, is the notion of Coordinated Behavior (CB), which captures suspicious or exceptional similarity in user actions, independent of their authenticity \cite{Nizzoli2021}.

Research has largely focused on detecting coordinated and inauthentic behavior \cite{pacheco2021uncovering, luceri2023unmasking, nwala2023language, Loru2026}. However, beyond detection, a fundamental question remains largely unexplored: to what extent do coordinated accounts actually influence the diffusion of information?

This question is particularly relevant in the current digital landscape, where the rise of Large Language Models (LLMs) has made artificially managed accounts increasingly realistic, scalable, and difficult to distinguish from genuine users. As these technologies lower the cost of generating convincing content and interacting at scale, understanding the actual influence exerted by coordinated actors becomes increasingly important for assessing their potential impact on online information ecosystems \cite{Kupferschmidt2026}.

A previous work has introduced a possible mathematical framework to address this problem \cite{DiMarco2025}, but it captures only a limited portion of its overall complexity, partly due to the challenges associated with obtaining detailed information cascade data.
In this work, we extend this line of research by both refining the mathematical framework and systematically evaluating the results produced by different approaches to measuring influence.

First, we develop a new optimal algorithm for computing maximum influence in directed tree structures, with particular relevance to cascades involving coordinated behavior.
Second, given the difficulty of accessing cascade data, we propose an alternative heuristic framework that operates directly on retweet networks populated by coordinated accounts.
The combination of these two approaches provides complementary perspectives, enabling a more comprehensive assessment of influence and highlighting both similarities and differences across modelling strategies.

To evaluate our framework, we rely on two datasets. The first consists of retweet cascades associated with posts published during the 2019 UK elections, providing a setting in which information diffusion can be directly observed \cite{Nizzoli2021}. The second comprises retweet networks derived from multiple coordinated information operations conducted across different countries, capturing large-scale patterns of coordinated activity in diverse contexts \cite{dataset_menczer}.

While the optimal algorithm does not identify a meaningful impact of CB in information cascades related to UK elections, a different picture emerges when using the retweet network setting, where campaigns from different countries exhibit varying levels of influence, suggesting substantial heterogeneity across campaigns, with patterns that in some cases align with country-level groupings.

To better understand whether these differences stem from methodological constraints or from the nature of the datasets, we further extend our analysis by constructing cascades directly from the retweet networks and reapplying the original framework. The results show that, also in this setting, the impact of CIB remains comparable to that observed in the retweet network analysis, pointing to a strong consistency between the two approaches.

Our framework provides a consistent and robust assessment of the role of coordinated actors in information diffusion. By bridging cascade-based and network-based perspectives, it offers a unified approach to better understand and quantify the influence of coordinated behavior across different contexts.

\section{Related works}
A central phenomenon in online ecosystems is the rise of coordinated activity \cite{pacheco2021uncovering, Tardelli2024temporal, Cinelli2022}, which frequently involves either automated agents or groups acting in a strategically organized manner \cite{hristakieva2022spread, weber2021amplifying}. Tools such as bots \cite{ferrara2016rise, Cresci2020decade}, sockpuppet accounts, and troll infrastructures \cite{Cheng2017anyone} are commonly deployed to mimic genuine user interactions, artificially inflate engagement, and shape perceived public opinion \cite{badawy2018analyzing,Shu2020,Starbird2019disinformation}. These activities, broadly categorized as coordinated inauthentic behavior (CIB), are often designed to promote specific narratives, marginalize dissenting voices, or fabricate the appearance of widespread grassroots endorsement—commonly referred to as astroturfing \cite{stella2018bots,luceri2019red, keller2020political, schoch2022coordination, Ratkiewicz2021detecting}. By producing large volumes of temporally aligned actions across multiple accounts, such campaigns can exploit platform algorithms, overwhelm information spaces, and influence users’ perceptions of credibility and popularity \cite{ratkiewicz2011truthy,badawy2018analyzing, Cinelli2019information}.

At the same time, not all coordinated patterns stem from manipulation. Online platforms also host forms of organic synchronization, where users independently converge toward similar behaviors without explicit coordination \cite{bennett2012logic,weller2013twitter,Centola2018}. Phenomena such as viral trends, memes, and hashtag cascades often generate aligned attention dynamics, typically driven by emotional contagion or shared identities rather than deliberate strategy \cite{castillo2014characterizing,pacheco2021uncovering}. Distinguishing between intentional coordination and emergent collective behavior is therefore essential for understanding how influence propagates in digital systems.

Recent developments in generative Artificial Intelligence have further expanded this landscape. In particular, Large Language Models (LLMs) are increasingly used to simulate group interactions and collective dynamics. Studies suggest that interacting LLM agents can spontaneously develop shared conventions and biases \cite{Ashery2025}, with group size playing a critical role in shaping these dynamics \cite{Flint2025}. At the same time, concerns have been raised about the potential misuse of coordinated AI agents in influence operations, which could intensify risks to democratic information environments \cite{Schroeder2025}. Notably, evidence shows that autonomous LLM-based agents are capable of reproducing coordination patterns similar to those observed in real-world information campaigns, even in the absence of direct human control \cite{Orland2025}.

In this landscape, identifying coordinated behavior poses significant methodological challenges \cite{alothali2018detecting, Cresci2020decade, Pant2025beyond}. Current approaches rely on a combination of techniques, including similarity analysis, temporal alignment, content duplication, and network-based indicators, to uncover coordinated entities across platforms \cite{DiMarco2025, Tardelli2024temporal, Minici2025iohunter, Cinus2025exposing}. In politically sensitive or toxic contexts, coordinated actors often display burst-like activity patterns during critical events, simultaneously engaging with the same content while evading detection mechanisms \cite{Cinelli2021dynamics, Loru2024, Cresci2019capability}. Although such campaigns are typically short in duration, they can significantly affect narrative construction, agenda-setting processes, and the emotional tone of online discussions.

For a more in-depth perspective on coordinated inauthentic behavior, readers may consult dedicated survey works such as \cite{CIB_Survey}.

\section{Data and Methods}
\subsection{Dataset}\label{sec:data}
To evaluate our framework, we rely on two complementary data sources: one for information cascades and one for retweet networks. We describe them in the following paragraphs.

\paragraph{Information cascades}
Our dataset is based on a collection of tweets related to the online debate surrounding the 2019 United Kingdom general election, originally presented in~\cite{Nizzoli2021}. Data were collected using the official Twitter API through a combination of hashtag-based queries (e.g., \#GeneralElection19, \#VoteLabour, \#VoteConservative) and timeline downloads of political parties and their leaders. The collection spans one month, from November 12 to December 12, 2019, and includes 11,264,820 tweets posted by 1,179,659 distinct users. The dataset is publicly available at \url{https://doi.org/10.5281/zenodo.4647893}.

The identification of coordinated accounts follows the methodology proposed in~\cite{Nizzoli2021}. The procedure consists of: (i) selecting influential users (super-spreaders), (ii) defining a similarity measure between users (e.g., cosine similarity), (iii) constructing a user similarity network via pairwise comparisons, (iv) filtering this network, and (v) applying clustering techniques. This process yields clusters of users associated with a continuous coordination score. In subsequent work~\cite{cinelli2022coordinated}, coordinated accounts are identified by retaining users connected by the top 1\% of similarity scores.

We employ information cascades from~\cite{cinelli2022coordinated}, which were constructed using a heuristic that we briefly summarize here. First, we observe that the impossibility of reconstructing information cascades exactly is due entirely to limitations of the Twitter API, where retweet data do not directly encode the true diffusion paths. In particular, if a user $j$ retweets a post by user $i$, and another user $k$ retweets the same content through $j$, the API records edges $(i \rightarrow j)$ and $(i \rightarrow k)$, resulting in a star-like structure that does not reflect the actual propagation process.

To address this issue, cascades are reconstructed by leveraging the underlying followers network. The approach assumes that users are more likely to retweet content from the most recent retweeter they follow. Under this assumption, the previous example would be reconstructed as $(i \rightarrow j)$ and $(j \rightarrow k)$, provided that $k$ follows $j$. When no such relationship exists, the node is treated as disconnected. As a result, each cascade may consist of a forest of directed trees rather than a single connected component. 
From these cascades, we assign a binary label to each node indicating whether the corresponding user is identified as participating in coordinated behavior according to the adopted detection algorithm. We emphasize that this labeling captures coordinated behavior (CB) rather than coordinated inauthentic behavior (CIB), as the available data do not allow the authenticity of the accounts to be established.

Applying this reconstruction procedure to 49,331 tweets yields an equal number of information cascades. Summary statistics of the resulting dataset are reported in Table~\ref{tab:twitter_trees}.
\begin{table}[!ht]
\centering
\begin{tabular}{lrrrrrr}
  \toprule
       Number of cascades & 49331 \\
\midrule 
       Minimum number of nodes & 1 \\
       Maximum number of nodes & 9066 \\
       Minimum number of coordinated accounts & 0 \\
       Maximum number of coordinated accounts & 236 \\
\bottomrule
    \end{tabular}
    \caption{Data breakdown of the cascades collected from Twitter.}
    \label{tab:twitter_trees}
\end{table}
\paragraph{Retweet networks}
We use a publicly available Information Operations dataset \cite{dataset_menczer} available at \url{https://zenodo.org/records/14189193}. 

The dataset consists of a collection of coordinated campaigns identified and released by a major social media platform as part of its transparency efforts, following the takedown of accounts involved in state-backed information operations. These campaigns span multiple countries and contexts, and include accounts that have been verified as participating in coordinated and inauthentic behavior. Importantly, the dataset does not only capture isolated posts, but reconstructs the activity timelines of the involved accounts, allowing for a detailed analysis of their behavior over time.

A key feature of this dataset is the inclusion of a carefully curated control group. For each campaign, control accounts are selected among users who engaged with similar topics—identified through shared hashtags and temporal overlap—but are not associated with coordinated activity. This design enables a meaningful comparison between coordinated and organic behavior, capturing both similarities in content and differences in interaction patterns. Moreover, the dataset includes not only posts directly related to the campaigns, but also additional activity from the same users, providing a broader view of their behavior.

For our purposes, we consider each campaign separately. For each dataset, we construct a network $G$ in which nodes are users and $(u,v)\in E$ if $v$ reposts a post from $u$.  
To retain only significant connections and ensure networks of manageable size, we apply the disparity filter \cite{Serrano2009} with $\alpha = 0.05$, and consider only the largest weak connected component for subsequent analysis.  
We then associate a binary label to the vertices of the networks, indicating whether it belongs to a coordinated group. Recall that, in this case, this label corresponds to platform-identified information-operation accounts.

In the analysis that follows, we focus on networks that, after this procedure, contain at least 100 nodes and at least one coordinated account. Table \ref{tab:summary_networks} shows summary statistics of the selected networks.
\begin{table}[ht]
\centering
\begin{tabular}{lrrrrrr}
  \toprule
Dataset & $\lvert V \rvert$ & $\lvert E \rvert$ & $<d>$ & $\lvert V_1 \rvert $ & $\frac{\lvert V_1 \rvert}{\lvert V \rvert}$ \\ 
  \midrule
  Uae & 5875 & 11835  & 4.03 & 2148 & 0.37 \\ 
  Thailand & 549 & 717  & 2.61 &  84 & 0.15 \\ 
  Spain & 275 & 634  & 4.61 & 140 & 0.51 \\ 
  Russia 5 & 4311 & 5463  & 2.53 &  12 & 0.00 \\ 
  Russia 4 & 10141 & 12865  & 2.54 &   7 & 0.00 \\ 
  Russia 2 & 4291 & 6032  & 2.81 &  61 & 0.01 \\ 
  Russia 1 & 5579 & 15541  & 5.57 & 1354 & 0.24 \\ 
  Quatar & 7144 & 10554  & 2.95 &  20 & 0.00 \\ 
  Iran 6 & 4967 & 6498 & 2.62 &  67 & 0.01 \\ 
  Iran 4 & 2443 & 3017  & 2.47 & 269 & 0.11 \\ 
  Iran 3 & 7355 & 10292  & 2.80 & 141 & 0.02 \\ 
  Iran 2 & 3309 & 5298  & 3.20 &  34 & 0.01 \\ 
  Iran 1 & 930 & 1695  & 3.65 & 158 & 0.17 \\ 
  Ghana nigeria & 246 & 280  & 2.28 &  37 & 0.15 \\ 
  Egypt uae & 415 & 940 & 4.53 & 165 & 0.40 \\ 
  Ecuador & 5875 & 10371  & 3.53 & 304 & 0.05 \\ 
  Cuba & 11512 & 38523  & 6.69 & 422 & 0.04 \\ 
  China 1 & 11765 & 15177  & 2.58 & 150 & 0.01 \\ 
  Catalonia & 897 & 1281  & 2.86 &  25 & 0.03 \\ 
   \bottomrule
\end{tabular}
\caption{Summary statistics of the networks considered in our analysis. $<d>$ denotes the mean degree of the network and $\lvert V_1 \rvert$ the number of coordinated accounts.}
\label{tab:summary_networks}
\end{table}
\subsection{Notation and Problem Definition}\label{sec:notation}
In this section, we introduce the notation and methodological framework used throughout the paper. Our goal is to retrospectively quantify the influence exerted by coordinated inauthentic behavior (CIB) in the dissemination of content. To this end, we adopt two complementary perspectives: a cascade-based approach, focusing on individual diffusion trees, and a general network-based approach, operating on aggregated retweet networks.

\paragraph{Common Notation}
We consider a directed graph $G = (V, E)$, where $V$ is the set of nodes and $E$ is the set of edges. For a node $v \in V$, we denote its out-neighborhood (excluding $v$ itself) as $N(v)$ and its out-degree as $d(v)$.

Each node $v$ is associated with a binary label $l(v) \in \{0,1\}$. We refer to $v$ as a $1$-node if $l(v) = 1$, and as a $0$-node otherwise. The labeling induces a partition of the node set $V = V_1^l \cup V_0^l$, where $V_i^l = \{ v \in V \mid l(v) = i \}$. The number of $1$-nodes is denoted by $k_l = |V_1^l|$. In our considered setting, the labeling $l$ captures whether a node (i.e. user) belongs to a coordinated group.

\paragraph{Directed Trees}
We first consider a directed tree $T = (V, E)$, representing the diffusion of a single post in a social media platform (e.g., Twitter/X). Each tree is rooted at a source node $r$, and we denote its size as $|T| = |V|$. 

To quantify the influence of coordinated accounts, we define, for each node $v$, the quantity $d_0^l(v) = |\{ w \in N(v) \mid l(w) = 0 \}|$, that is, the number of non-coordinated nodes in the out-neighborhood of $v$. When the context is clear, we omit the superscript $l$. This measure captures the number of non-coordinated users directly influenced by a coordinated node.

The overall influence of a configuration $(T,l)$ is then defined as
\begin{equation}\label{eq:influence}
    I(T,l) = \sum_{v \in V_1^l} d_0(v),
\end{equation}
i.e., the total number of non-coordinated nodes directly reached by coordinated ones.

We consider the following optimization problem:

\begin{tcolorbox}[colback=gray!20, colframe=gray!50, sharp corners]
\noindent {\bf Constrained Influence Maximization Problem} $(CIMP(T,k))$:

\smallskip

\noindent Given a directed tree $T$ and an integer $0 \leq k \leq |T|$, find a labeling $\bar{l}$ such that $|V_1^{\bar{l}}| = k$ and
$I(T,\bar{l}) = \max_{l:\, |V_1^l| = k} I(T,l)$.
\end{tcolorbox}

The solution $\bar{l}$ is referred to as an \emph{optimal labeling}, and the corresponding value $I(T,\bar{l})$ as the \emph{maximal influence}. Note that multiple optimal labelings may exist \cite{DiMarco2025, DiMarco2025_Constrained_Maximization_Tree}. Solving this problem provides the theoretical maximum of \eqref{eq:influence}. This value can then be used to normalize the influence achieved by the observed coordinated accounts, yielding a ratio bounded by $1$.

\paragraph{Retweet Networks}\label{sec:model}
We now consider networks that aggregate multiple diffusion processes.
In this setting, to model information spreading, we adopt the independent cascade model, a well-established framework in the literature \cite{Kempe2003,Kempe2005,Mossel2007}. Starting from a set of seed nodes $S \subseteq V$, the process unfolds in discrete time steps. When a node $v$ becomes active for the first time, it has a single chance to activate each neighbor $w \in N(v)$ with probability $p_{vw}$, which represents the strength of the social tie between $v$ and $w$. If successful, $w$ becomes active at the next time step. The process continues until no further activations occur. In our analysis, we assume a uniform infection probability $p$ for all edges. This choice allows us to isolate the effect of the network topology and the positioning of coordinated accounts, without introducing additional assumptions on the strength of individual interactions. Nevertheless, the framework naturally extends to edge-specific infection probabilities whenever such information is available.

Let $A_\infty(S)$ denote the set of active nodes at the end of the process. The influence of a seed set $S$ is defined as
\[
I(S)=\mathbb{E}[|A_\infty(S)|].
\]
that is, the expected number of activated nodes, where $\mathbb{E}$ denotes expectation. Since exact computation is generally intractable, $I(S)$ is typically estimated via Monte Carlo simulations \cite{Kempe2003,Kempe2005}. We follow the same approach here, estimating it through computer simulations.

To evaluate the influence of coordinated accounts, we set $S = V_1^l$, i.e., we measure the long-term influence obtained by $1-$nodes.

However, evaluating influence solely on observed coordinated nodes is not sufficient to fully assess their impact. A meaningful analysis requires comparison with alternative placement strategies. Since identifying the optimal seed set is NP-hard \cite{Kempe2005}, we benchmark observed configurations against two heuristics:
\begin{itemize}
    \item[\textbf{K-shell decomposition:}] Nodes are assigned a coreness index $k_s$ through an iterative pruning process \cite{Seidman1983,Carmi2007}. Given $k$ coordinated accounts, we compare the influence achieved by them with that obtained by selecting the $k$ nodes with the highest $k_s$. This method has been used to identify influential spreaders \cite{Kitsak2010} in networks.
    \item[\textbf{Random selection:}] Nodes are selected uniformly at random, providing a null baseline.
\end{itemize}
To provide a concrete example, Figure \ref{fig:simulations} shows the results of the independent cascade model applied to the {\it Catalonia} network, comparing influences across the three different sets of seed nodes. 
\begin{figure}[!ht]
    \centering
    \includegraphics[width=\linewidth]{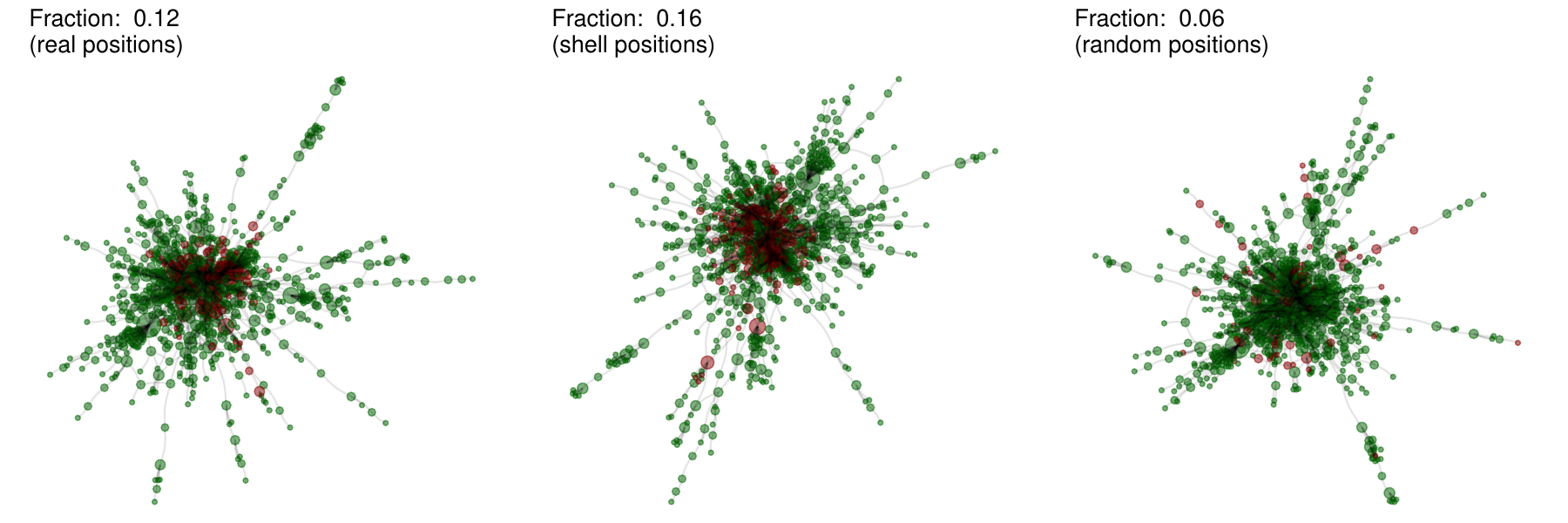}
    \caption{Results of the cascade models applied to three different sets of seed nodes: the first uses the positions of real IO accounts, while the second and third correspond to shell and random positions, respectively. The infected nodes are highlighted in red, and their fraction is reported.}
    \label{fig:simulations}
\end{figure}
\subsection{Dynamic programming solution of $CIMP(T,k)$}
In this section, we show that $CIMP(T,k)$ can be solved in polynomial time with respect to the number of nodes of $T$ through a dynamic programming approach. 
Let $T$ be a directed tree and $v$ one of its nodes. We indicate as $T_v$ the (directed) subtree of $T$ whose root is $v$ and as $l_v$ a label of the subtree rooted at $v$.
Finally, we denote with $CIMP_0(T_v,k)$ [resp. $CIMP_1(T_v,k)$] a solution $l$ of $CIMP(T_v,k)$ with the constraints that $l(v)=0$ [resp. $l(v)=1$] and $k\leq |T_v|$. 

\begin{Prop}\label{prop:maxchildren}
    Let $T$ be a tree rooted at $v$ and $v_1,\dots,v_d$ children of $v$. Each optimal labeling $l_v$ of $T_v$ is formed by $l_{v_1}, \dots, l_{v_d}$, labelings of $T_{v_1},\dots,T_{v_d}$ that are optimal.  
\end{Prop}

% \begin{Prop}\label{prop:maxchildren}
%     Let $T$ be a tree rooted at $v$, with $v_1,\dots,v_n$ children of $v$ and $k = k_1 + \ldots k_n$. The following statements hold: 
%     \begin{description} 
%     \item{$(i)$} it exists a non-maximal labeling $l_v$ with $\lvert l_v \rvert = k$ such that all $l_{v_1}, \ldots , l_{v_n}$, with $\lvert l_{v_i}\rvert = k_i$, are maximal;  
%     \item{$(ii)$} it exists a labeling $l_v$ that is maximal, and such that the labelings $l_{v_1}, \ldots , l_{v_n}$, with $\lvert l_{v_i}\rvert = k_i$, are also maximal.  
% \end{description}
% \end{Prop}
%An example of $(i)$ is in Fig.~\ref{fig:kcounterexample}.

\proof 
Let us proceed by contradiction, assuming that the optimal labeling $l$ contains, w.l.g., the label $l_{v_1}$ that is not optimal. The following two cases arise according to the value of $l(v)$: 

\begin{itemize}
\item if $l(v) = 0$, we have $I(T_v,l_v) = I(T_{v_1},l_{v_1}) + \ldots + I(T_{v_d}, l_{v_d})$. Therefore, the non-optimal labeling  $l_{v_1}$ can be substituted by an optimal one, thus increasing the total influence, against the optimality of $l_v$.
\item if $l(v)=1$ again $l_{v_1}$ can be changed into a maximal labeling $l_{v_1}'$. However, this does not guarantee, in general, that the total influence increases. In fact, the following situation may be present: $l(v)=1$, $l(v_1)=0$, $l'(v_1)=1$. In this case, changing from $l_{v_1}$ to $l_{v_1}'$ decreases the influence obtained by $v$ by one. However, by hypothesis we have $I(T_{v_1},l_{v_1}) < I(T_{v_1},l_{v_1}')$, so the new influence on $T_v$ is equal or greater than the previous one. The label $l$ is already optimal. \qed     
\end{itemize}

Property~\ref{prop:maxchildren} enables us to design a dynamic programming strategy. 

The strategy relies on the computation of $S_{0,v}(i)$ and $S_{1,v}(i)$ that, for each node $v$ of $T$ and for each $i = 0,1,\ldots,k$, store the values of the maximal influences associated with their related subtree $T_v$ when the label of $v$ is either $0$ or $1$ and the number of coordinated nodes is $i$.
In particular, the $i$-th entries $S_{0,v}(i)$ (resp. $S_{1,v}(i)$) is defined as $I(T_v,l_v)$, with $l_v=CIMP_0(T_v,i)$ (resp. $l_v=CIMP_1(T_v,i)$). 
We associate to each node a vector of length $k+1$. If $i> |T_v|$ we set $S_{0,v}(i)=S_{1,v}(i)=-\infty$.

If $v$ is a leaf, the only feasible configurations are the following:
\[
S_{0,v}(0)=0, \qquad S_{1,v}(1)=0.
\]
All other entries are infeasible and are therefore set to $-\infty$. In other
words, for a leaf $v$,
\[
S_{0,v}(i)=
\begin{cases}
0 & \text{if } i=0,\\
-\infty & \text{otherwise,}
\end{cases}
\qquad
S_{1,v}(i)=
\begin{cases}
0 & \text{if } i=1,\\
-\infty & \text{otherwise.}
\end{cases}
\]

If $v$ is the root of $T$, then $\max\{S_{0,v}(i),S_{1,v}(i)\}$ stores the maximal influence of one labeling of $CIMP(T,i)$. Therefore, its $k$-th entry provides the required solution to $CIMP(T,k)$. 

Relying on Property~\ref{prop:maxchildren}, the following Lemma \ref{lem:solution}, proved in \cite{DiMarco2025_Constrained_Maximization_Tree}, shows that these two vectors can be defined in terms of the children $v_1,\dots,v_d$ of $v$ as follows

% \begin{equation}\label{eq:s0s1}
% \begin{cases}
%       S_{0,v}(i) = \max\limits_{i_1+i_2+\dots+i_t=i}\{ \max\{S_{0,v_1}(i_1),S_{1,v_1}(i_1)\}+\dots + \max\{S_{0,v_d}(i_d),S_{1,v_d}(i_d)\} \} \\
%    S_{1,v}(i) = \max\limits_{i_1+i_2+\dots+i_t=i-1}\{ \max\{S_{0,v_1}(i_1),S_{1,v_1}(i_1)\}+\dots + \max\{S_{0,v_d}(i_t),V_{1,v_d}(i_d)\} + d_0(v)\} 
% \end{cases}
% \end{equation}

\begin{equation}\label{eq:s0s1}
\begin{cases}
S_{0,v}(i)=
\displaystyle
\max_{i_1+\cdots+i_d=i}
\sum_{h=1}^{d}
\max\{S_{0,v_h}(i_h),S_{1,v_h}(i_h)\},
\\[1.2em]
S_{1,v}(i)=
\displaystyle
\max_{i_1+\cdots+i_d=i-1}
\sum_{h=1}^{d}
\max\{S_{0,v_h}(i_h)+1,S_{1,v_h}(i_h)\}.
\end{cases}
\end{equation}

\begin{lem}\label{lem:solution}\cite{DiMarco2025_Constrained_Maximization_Tree}
For each node $v$ of a directed tree $T$, the values $S_{0,v}(i)$ and $S_{1,v}(i)$ computed according to equations (\ref{eq:s0s1}), store the influences of the labels that are solutions of $CIMP_0(T_v,i)$ and  $CIMP_1(T_v,i)$, respectively. 
\end{lem}

The intuition behind Eq.~\eqref{eq:s0s1} is the following. To compute $S_{0,v}(i)$, we distribute the $i$ coordinated nodes among the subtrees rooted at the children of $v$. Each partition $i_1+\cdots+i_d=i$ specifies how many coordinated nodes are assigned to each subtree. For a fixed partition, each child $v_j$ independently contributes its best possible influence, namely
$\max\{S_{0,v_j}(i_j),, S_{1,v_j}(i_j)\}$.

If $S_{0,v_j}(i_j) > S_{1,v_j}(i_j)$, the optimal solution labels $v_j$ as a $0$-node; otherwise, it is preferable to label $v_j$ as a $1$-node. The value of $S_{0,v}(i)$ is then obtained by selecting the partition of $i$ that maximizes the sum of these optimal contributions.

The computation of $S_{1,v}(i)$ follows the same principle, with two differences. First, since $v$ itself is labeled as a coordinated node, only $i-1$ coordinated nodes remain to be distributed among its children. Second, for each child $v_h$, assigning label 0 to $v_h$ contributes one additional unit of influence through the edge $(v,v_h)$. This is captured by the term $S_{0,v_h}(i_h)+1$, while $S_{1,v_h}(i_h)$ is used when the child root is labeled as a $1-$node.

% The computation of $S_{1,v}(i)$ follows the same principle, with two differences. First, since $v$ itself is labeled as a coordinated node, only $i-1$ coordinated nodes remain to be distributed among its children. Second, the contribution $d_0(v)$ is added to account for the influence generated by $v$, which depends on how many of its children are labeled as $0$-nodes in the chosen partition.

To clarify these aspects, let's consider the following example.

\begin{Es}\label{ex:computes0s1}
Let us consider the directed tree in Fig.~\ref{fig:esempio1}. Setting $k=3$, we list the vectors $S_{0,i}$ and $S_{1,i}$ for $i = \{ v_1, \ldots , v_{10}\}$. 
First of all, note that:

\begin{enumerate}
\item $S_{0,v} = \left(0, -\infty, -\infty, -\infty \right)$ and $S_{1,v} = \left(-\infty, 0, -\infty, -\infty \right)$ for $v = \{v_5,v_6,v_7,v_8,v_9,v_{10}\}$;
    \item $S_{0,v_3} = S_{0,v_4} = (0,0,0,0)$;
    \item $S_{1,v_3}=(-\infty,3,2,1)$;
    \item $S_{1,v_4}=(-\infty,2,1,0)$.
\end{enumerate}

To understand how vectors are constructed, we initially consider node $v_2$.
For example, $S_{0,v_2}(1)=3$ can be obtained by choosing the sequence $S_{1,v_3}(1)=3, S_{0,v_4}(0)=0, S_{0,v_5}(0)=0$ related to the vectors of its children. With a similar rationale, we get $S_{0,v_2}(2)=S_{0,v_2}(3)=5$ and $S_{1,v_2}=(0,3,5,6)$. 

Let's consider the case $k = 3$. We obtain $S_{1,v_2}(3)=6$ by summing the sequence of values $S_{1,v_3}(1)=3$, $S_{1,v_4}(1)=2$, $S_{1,v_5}(0)=0$. Finally, we reach node $v_1$, that gets $S_{0,v_1}=(0,3,5,6)$ and $S_{1,v_1}=(0,1,4,6)$. Therefore, this approach provides the two optimal labels depicted in Fig.~\ref{fig:esempio1} with $6$ as maximal influence.

\begin{figure}[!ht]
    \centering
    \includegraphics[width=0.8\linewidth]{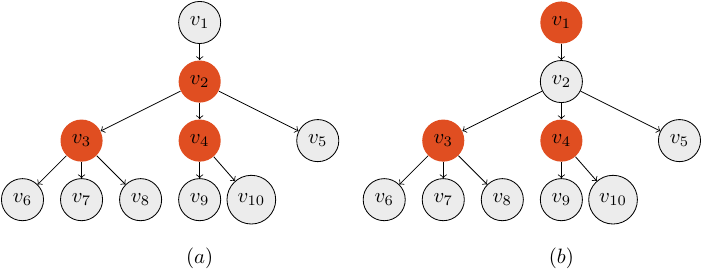}
    \caption{The two solutions of $CIMP(T,3)$. The labels of the nodes can be obtained from the computation of maximum values of the vectors $S_{0,v_1}$ and $S_{1,v_1}$, being node $v_1$ the root of the tree.}
    \label{fig:esempio1}
\end{figure}

\end{Es}

% \begin{Thm}
% \label{lem:solution}
% Given a directed tree $T$ rooted at $v$ and an integer $k$, it holds: 
% \begin{description}
% \item{$(i)$} the values $S_{0,v}(i)$ and $S_{1,v}(i)$, store the influences of the labels that are solutions of $CIMP_0(T_v,i)$ and  $CIMP_1(T_v,i)$, respectively;
% \item{$(ii)$} the problem $CIMP(T,k)$ can be solved in $O(|V|(k+\bar{d})^{\min\{k,\bar{d}\}})$ time, with $\bar{d}=\max_{v\in V}{d(v)}$.
% \end{description}
% \end{Thm}

The previous equations (\ref{eq:s0s1}) allow to compute the optimal label of $T$ in $O(|T|(k+\bar{d})^{\min\{k,\bar{d}\}})$ time, where $\bar{d}$ is the maximum degree of the nodes of $T$ \cite{DiMarco2025_Constrained_Maximization_Tree}. 
As one can easily argue, the exponential part is related to the computation of the $i$-th entries of the two vectors $S_{0,v}$ and $S_{1,v}$ that span over all the integer partitions of $i$ as a sum of $i_1,\dots,i_d$, and whose number grows exponentially in the parameter $i$.

To lower such computational cost, before proceeding with the vectors definition, we modify the tree $T$ producing an equivalent binary tree $T'$. This process replaces each node $w$ with degree $d :=  deg(w)>2$ of $T$ with a binary subtree obtained by arranging $d-1$ new (pseudo-)nodes $w^{1},\dots,w^{d-1}$ in a downward sloping line, and attaching to each of them, as a second child, one children of $w$. To clarify this approach Fig.~\ref{fig:binarization} shows an example. 

\begin{figure}[!ht]
    \centering
    \includegraphics[width=0.8\linewidth]{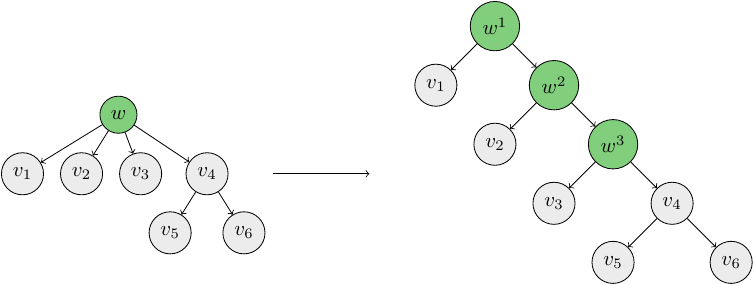}
    \caption{On the left, the tree $T$ has the node $w$ of degree $4$; on the right, the binary tree $T'$ obtained by substituting the node $w$ of $T$ by three new nodes $w_1$, $w_2$, and $w_3$. The children of the three nodes are the same of the node $w$ of $T$.}
    \label{fig:binarization}
\end{figure}

Note that the procedure to modify $T$ in $T'$ runs in $O(\lvert V \vert)$ time and $\lvert V(T') \rvert \leq 2\lvert V(T) \rvert$.
However, we highlight that new recurrence equations are needed for creating the vectors $S'_{0,w^i}$ and $S'_{1,w^i}$ relative to the nodes $w^i, i = 1, \ldots d(w)-1$. In fact, to maintain coherence, all those nodes must have the same label as the label of $w$ in $T$. 
Therefore, each node $w^i$ inherits its label from the child $w^{i+1}$. Suppose that the non-pseudo child of $w_i$ is $v_i$. The following new recurring equations defining the vectors $S'_{0,w_i}$ and $S'_{1,w_i}$, for $i = 1, \ldots d-2$:

% \begin{equation}\label{eq:s0s1new}
% \begin{cases}
%       S'_{0,w_i}(j) = \max\limits_{j_1+j_2=j}\{ S'_{0,w_{i+1}}(j_1)+\max\{S'_{0,v_i}(j_2),S'_{1,v_i}(j_2)\} \} \\
%    S'_{1,w_i}(j) = \max\limits_{j_1+j_2=j}\{ S'_{1,w_{i+1}}(j_1)+ \max\{S'_{0,v_i}(j_2),V_{1,v_i}(j_2)\} + d_0(w_i)\}
% \end{cases}
% \end{equation}

\begin{equation}\label{eq:s0s1new}
\begin{cases}
S'_{0,w_i}(j)=
\displaystyle
\max_{j_1+j_2=j}
\left[
S'_{0,w_{i+1}}(j_1)
+
\max\{S'_{0,v_i}(j_2),S'_{1,v_i}(j_2)\}
\right],
\\[1.2em]
S'_{1,w_i}(j)=
\displaystyle
\max_{j_1+j_2=j}
\left[
S'_{1,w_{i+1}}(j_1)
+
\max\{S'_{0,v_i}(j_2)+1,S'_{1,v_i}(j_2)\}
\right].
\end{cases}
\end{equation}

The term $+1$ appears only in the case where the attached original child $v_i$ is labeled as a 0-node, since then the pseudo-node $w_i$, labeled as a 1-node, contributes one unit of influence through the edge $(w_i,v_i)$. Note that the computation of $S'_{1,w_i}$ requires that the two indexes $j_1$ and $j_2$ sum to $j$ since the contribution of the label $l(w_i)=l(w_{i+1})=1$ has been already considered in $S'_{1,w_{i+1}}$.
Note that the last pseudo-node $w_{d-1}$, which is attached to the two remaining original children, is computed using the standard binary recurrence.

\begin{Thm}\label{teo:complexity}
  Given a directed tree $T=(V,E)$ rooted at $v$ and an integer $0\leq k \leq n$, with $n=|T|$, the problem $CIMP(T,k)$ can be solved in $O(nk^2)$ time.
\end{Thm}
\proof
By Lemma~\ref{lem:solution} the value $\max\{S_{0,v}(k),S_{1,v}(k)\}$ gives the
influence of an optimal labeling for $CIMP(T,k)$, where $v$ is the root of
$T$. We first transform $T$ into the equivalent binary tree $T'$, which
contains at most $2n$ nodes; this process runs in $O(n)$ time.

For each node $v'\in T'$ and each $j\in\{0,\ldots,k\}$, the value
$S'_{0,v'}(j)$ is computed by considering all pairs $(j_1,j_2)$ such that
$j_1+j_2=j$. There are $j+1\leq k+1$ such pairs, and therefore one entry
can be computed in $O(k)$ time. The same argument applies to
$S'_{1,v'}(j)$. Since there are $k+1$ entries for each of the two vectors,
the cost per node is $O(k^2)$.
 As $T'$ contains $O(n)$ nodes, the total
running time is $O(nk^2)$.

\qed

We underline that Theorem~\ref{teo:complexity} greatly reduces the complexity by lowering the exponential dependence on the node degree with a quadratic dependence on $k$.
\begin{Es}
Let us consider the directed tree $T$ considered in Example 1. Figure~\ref{fig:final_example} shows the construction of its equivalent binary tree $T'$. We now proceed to compute the new vectors $S'_0$ and $S'_1$ for $k=3$. In $T'$, the node $v_2$ of degree $3$ is replaced by the two nodes $v^2_1$ and $v^2_2$. The same happens to the node $v_3$ that is replaced by the two nodes $v^3_1$ and $v^3_2$. 

Since the nodes $i = \{ v_4, \ldots , v_{10}\}$ of $T'$ are not affected by the binarization process, then their vectors are equal to those of $T$, computed in Example \ref{ex:computes0s1} through equations (\ref{eq:s0s1}). 

On the other hand, we list the vectors for the new nodes:

\begin{itemize}
    \item $ S_{0,v_3^2}=(0,0,0,0) $;
    \item $S_{1,v_3^2}=(0,2,1,0)$;
    \item $ S_{0,v_2^2}=(0,2,2,1)$;
    \item $S_{1,v_2^2}=(0,2,3,2)$;
    \item $ S_{0,v_3^1}=(0,0,0,0)$;
    \item $S_{1,v_3^1}=(0,3,2,1)$;
    \item $ S_{0,v_2^1}=(0,3,5,5) $;
    \item $S_{1,v_2^1}=(0,2,5,6)$.
\end{itemize}

Note that $S_{i,v_2^2}$ and $S_{i, v_3^2}$ are computed using \eqref{eq:s0s1} while the others are computed using \eqref{eq:s0s1new}. 

Finally, the vectors of $v_1$ are the same in $T'$ and $T$, as expected. We underline that the vectors of $v_2^1$ and $v^1_3$ in $T'$ equal those of $v_2$ and $v_3$ in $T$, respectively.  
\end{Es}

\begin{figure}[!ht]
    \centering
    \includegraphics[width=0.9\linewidth]{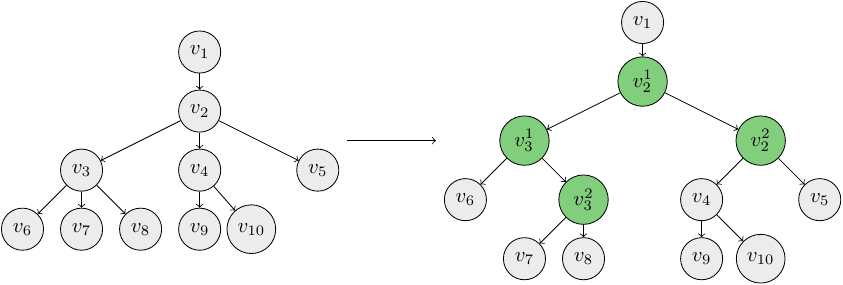}
    \caption{Binarization of the tree $T$ considered in Example \ref{ex:computes0s1}. The computation of the vectors $S'_0$ and $S'_1$ change only for the green (pseudo-)nodes $v_2^1, v_2^2, v_3^1, v_3^2$.}
    \label{fig:final_example}
\end{figure}

\section{Results}
\subsection{Influence on information cascades}
In this section, we apply our constrained algorithm to evaluate the influence achieved by coordinated accounts in real-world scenarios, specifically Twitter diffusion cascades. 
Given a cascade $c$ with $k$ coordinated accounts, we define the measure
\begin{equation}\label{eq:bounded_influence}
\rho_k(c)=\frac{I_{\mathrm{obs}}(c,k)}{I_{\mathrm{opt}}(c,k)}.
\end{equation}
where $I_{\mathrm{obs}}(c,k)$ denotes the influence achieved by the actual coordinated accounts in cascade $c$, and $I_{\mathrm{opt}}(c,k)$ denotes the optimal influence achievable by any placement of $k$ coordinated accounts in the same cascade.
This ratio provides a bounded measure of influence: the closer $\rho$ is to 1, the greater the impact of coordination. 

We apply our algorithm to the cascades obtained through the procedure described in Section~\ref{sec:data}, computing $\rho$ for each of them. To ensure both robustness and computational feasibility, we restrict the analysis to cascades containing at least $15$ nodes and at least one coordinated account. In addition, we exclude cascades containing a connected component with more than $4000$ nodes. These criteria yield a total of $4116$ cascades.
Since the reconstruction procedure may produce forests rather than single trees, each connected component is treated as an independent instance of $CIMP(T,k)$. The influence of a forest is then defined as the sum of the influences of its constituent trees.

Figure \ref{fig:comparison_fixed_k} presents the results of our analysis. 
\begin{figure}[!ht]
    \centering
    \includegraphics[width=0.9\linewidth]{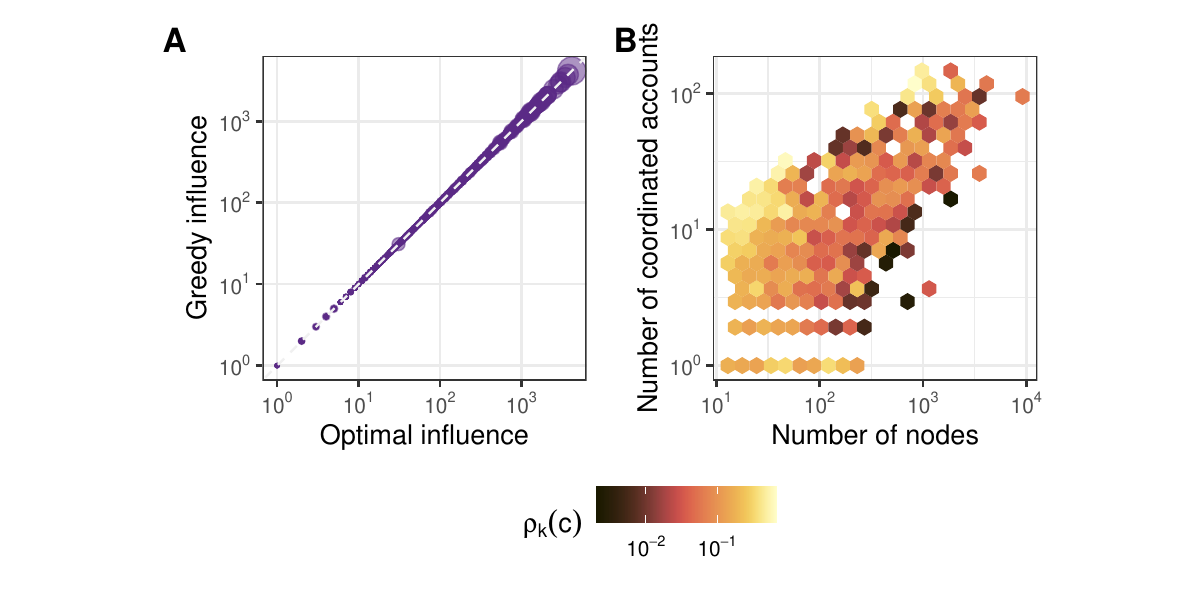}
    \caption{$(A)$ Relationship between optimal and greedy influence in Twitter/X cascades. The dashed line represents the identity relation. $(B)$ Bivariate distribution of the number of nodes and coordinated accounts in each Twitter/X cascade. Each hexagonal bin is colored according to the mean value of $\rho_k(c)$ for the cascades within the bin, using a logarithmic color scale.}
    \label{fig:comparison_fixed_k}
\end{figure}
We first aim to assess whether our algorithm achieves results comparable to those obtained by the heuristic introduced in \cite{DiMarco2025, DiMarco2025_Constrained_Maximization_Tree}.
Panel $A$ compares the optimal value (x-axis) with the greedy value (y-axis) obtained on the same tree, where $k$ is set equal to the number of coordinated accounts present in the cascade. The size of each point is proportional to the cascade size. The clustering of points near the diagonal indicates that the two algorithms yield similar values, despite the suboptimal nature of the greedy approach. However, the optimal algorithm is generally more time-efficient than the greedy approach, as shown in Section \ref{sec:comparison} of the appendix.

Panel $B$ shows the bivariate distribution of cascade size (number of nodes) and $k$ for each cascade. Each hexagonal bin is colored, using a logarithmic scale, according to the mean value of $\rho_k(c)$ for the cascades it contains. This approach ensures that only cascades of comparable size are grouped together, thereby reducing potential biases due to size heterogeneity. 

The figure indicates that high values of $\rho_k(c)$ are primarily observed in small cascades with a limited number of coordinated accounts. In contrast, as both the cascade size and $k$ increase, $\rho_k(c)$ rapidly decreases and stabilizes around $\approx 10^{-1}$, indicating that the achieved influence is often about $10\%$ of the optimal value. Although in some cases a relatively high influence is observed even in cascades with approximately $1000$ nodes, our results overall suggest a generally low efficiency of coordinated accounts, in agreement with previous findings \cite{DiMarco2025}.

\subsection{Influence on retweet networks}
In the previous section, we focused on information cascades. Given the difficulty of obtaining such data in real-world settings, we now turn to retweet networks, which we interpret as the outcome of multiple retweet cascades. As discussed in Section~\ref{sec:notation}, we assess the impact of CIB using an independent cascade model, measuring the expected number of infected users when CIB-selected nodes act as seeds.

We then compare this performance on real networks with two heuristic baselines, namely the $k$shell and random baselines introduced in Section~\ref{sec:notation}. While this comparison does not yield a strict bound on CIB's influence, it provides a meaningful benchmark by contrasting its behavior with both an established efficient strategy (seed-based) and a null model (random).

For clarity of presentation, the results in this section are obtained using a fixed infection probability $p = 0.3$ and only networks with at least $500$ nodes. The framework, however, naturally extends to other values of $p$, including edge-specific probabilities. Additional robustness analyses exploring these variations are reported in Section~\ref{sec:robustness_retweet_networks}, where we show that different values of $p$ yield consistent results. 

Figure \ref{fig:simulations_IO}$A$ shows the distribution of the fraction of infected users at the end of the simulation across heuristics and networks. We have repeated each simulation $20$ times to ensure sufficient observations.

\begin{figure}[!ht]
    \centering
    \includegraphics[width=0.8\linewidth]{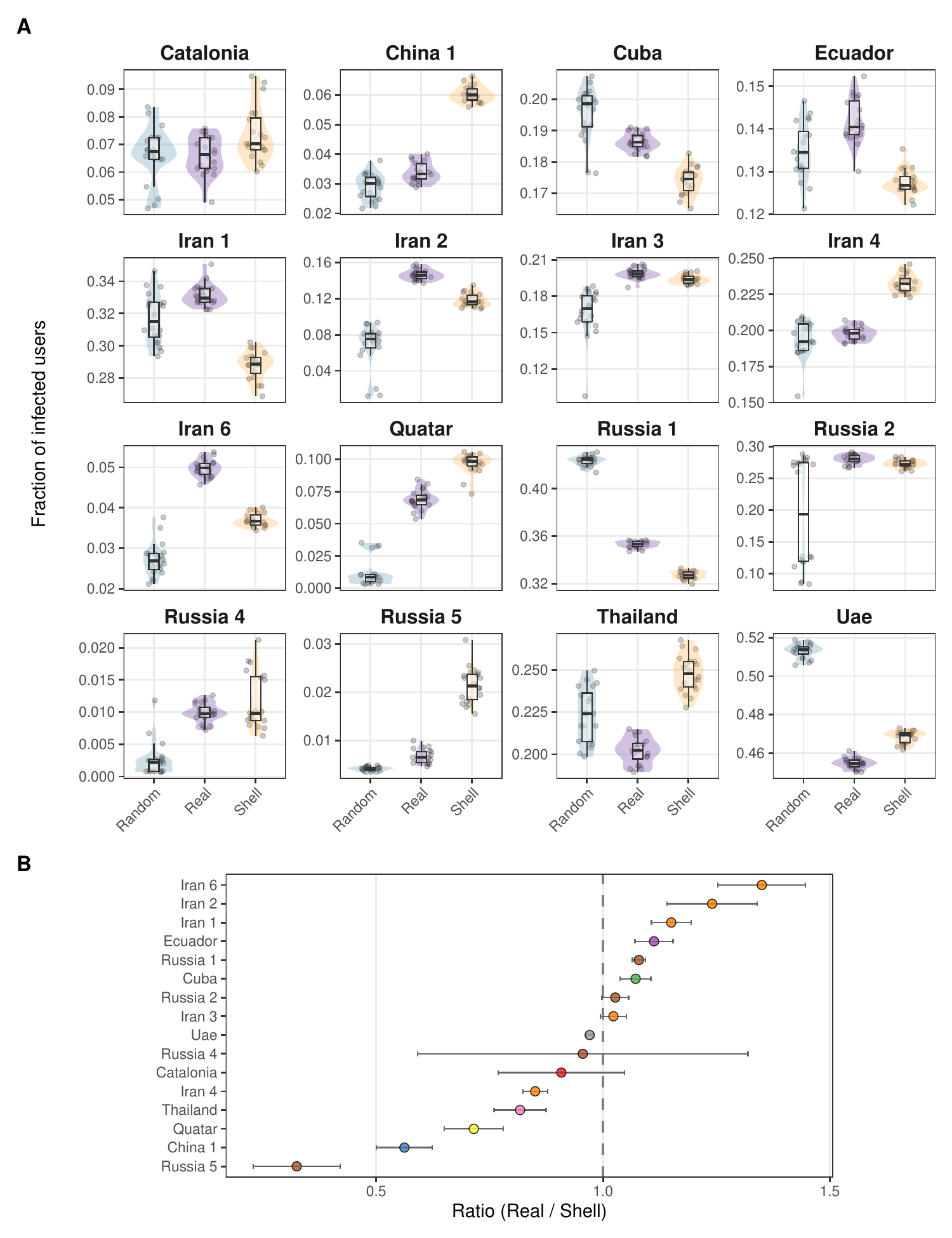}
    \caption{\textbf{(A)} Distribution of the fraction of infected nodes obtained from different seed node sets, computed over $20$ simulations for each set of seeds. \textbf{(B)} Mean and standard deviation of the ratio R~\eqref{eq:ratio} for each IO campaign. Points corresponding to the same country are shown in the same color.}
    \label{fig:simulations_IO}
\end{figure}

As shown, \textit{shell}-based heuristics generally achieve better performance, while the random baseline consistently yields the lowest results, albeit with some notable exceptions. Interestingly, CIB often outperforms the shell-based approach, for instance, in most Iran-related networks and in Ecuador. Even when this is not the case, CIB and \textit{shell}-based methods tend to exhibit similar performance.

Overall, these findings indicate that CIB actors are well-positioned within the retweet networks and, within this framework, may have substantial diffusion potential in the observed retweet networks.

These initial results also point to campaign-level heterogeneity, which we further inspect by comparing operations associated with different countries. To quantify it more systematically, we compute the average influence observed in each simulation and define the ratio
\begin{equation}\label{eq:ratio}
    R = \frac{I_{\text{real}}}{I_{\text{shell}}},
\end{equation}
which provides a numerical indicator of whether CIB outperforms or underperforms the shell-based approach. We recall that values of $R > 1$ indicate a higher influence of CIB compared to the (efficient) shell-based approach.

Figure~\ref{fig:simulations_IO}$B$ reports, for each IO operation, the mean and standard deviation of $R$, averaged over the $20$ simulations. Points corresponding to campaigns from the same country are shown in the same color. 

Consistent with previous observations, CIB nodes in most Iran-related campaigns tend to achieve higher influence, with similar patterns emerging for Ecuador and Cuba. In contrast, campaigns associated with Russia, China, and Thailand generally exhibit lower values of $R$. 

To assess whether these differences are simply driven by the number of coordinated accounts in each network, we computed the correlation between $R$ and $|V_1|$ across all networks. We found a weak, non-significant correlation ($r = 0.14$, $p = 0.61$) suggesting that the observed differences are unlikely to be driven solely by the number of coordinated users, and are consistent with the interpretation that network position plays an important role.

These results highlight a non-trivial geographical heterogeneity in the positioning and effectiveness of CIB actors. Such differences are particularly relevant, as they suggest that the potential impact of coordinated behavior is not uniform across contexts and depends on how these actors are embedded within the underlying network structure.

\subsection{Linking information cascades and retweet networks}
In the previous section, we presented results obtained using two distinct frameworks. However, it remains unclear whether the observed discrepancies arise from differences in the data or from the modeling frameworks themselves.
Therefore, it is important to relate these approaches and evaluate their performance on the same dataset.

To address this issue, we construct information cascades directly from the retweet networks associated with IO operations. This enables us to apply the cascade-based framework to the same data and assess whether it produces consistent results.

More precisely, for each network with at least $500$ nodes, we identify the node with the largest out-component, i.e., the set of nodes reachable via directed paths originating from that node. Starting from this node, we perform a breadth-first search (BFS) to extract a directed tree that approximates an information cascade of retweets.

For each cascade $c$ constructed in this way, we compute $\rho_k(c)$ (see Eq.~\eqref{eq:bounded_influence}), with $k$ equal to the number of coordinated accounts in the cascade, to assess its proximity to the optimal influence. We also note that one of the resulting cascades contains no coordinated accounts and is therefore excluded from the analysis.

The results are shown in Figure~\ref{fig:cascades_IO}, which illustrates the relationship between the cascade size and $\rho_k(c)$. The color of each point encodes the value of $k$.
\begin{figure}[!ht]
    \centering
    \includegraphics[width=0.7\linewidth]{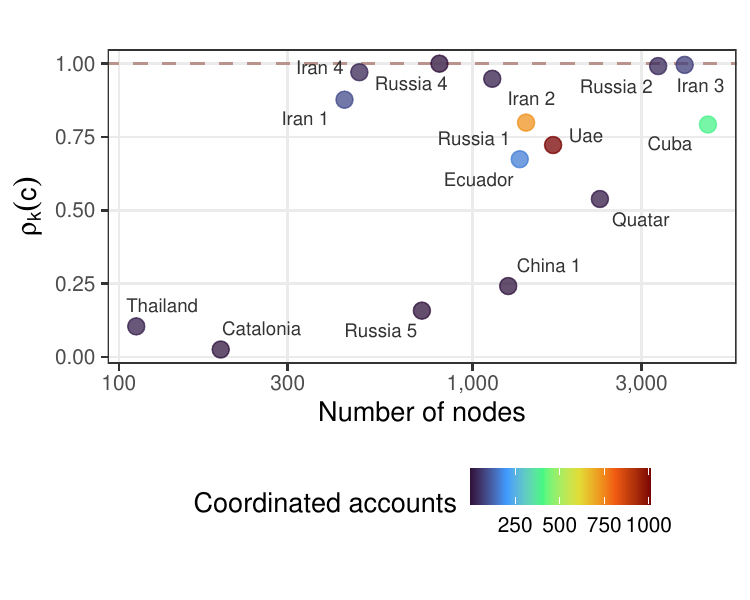}
    \caption{Impact of CIB accounts on spanning trees extracted from retweet networks. The results are consistent with those obtained from the independent cascade simulations on the corresponding retweet networks, providing further evidence of the robustness of the findings and the coherence between the cascade-based and network-based frameworks.}
    \label{fig:cascades_IO}
\end{figure}
As shown in the figure, with the exception of very small cascades, most cases achieve high values of $\rho_k(c)$, indicating performance close to the optimal solution. While the highest values are typically associated with cascades involving a small number of coordinated accounts, it is noteworthy that in several cases—such as UAE, Iran~2, Russia~1, Cuba, and Ecuador—high values are observed even for large $k$.
Notably, the networks that achieve the highest values of $\rho_k(c)$ are also those with the largest $R$ scores in Figure~\ref{fig:simulations_IO}(B), suggesting a strong consistency between the two approaches.

\section{Conclusions}
In this paper, we introduced two complementary frameworks for the post-hoc evaluation of coordinated behavior in online social networks. On the theoretical side, we formulated the problem of quantifying the influence of coordinated accounts in information cascades as a constrained optimization problem on directed trees and proposed a polynomial-time dynamic programming algorithm to compute its optimal solution. On the empirical side, we introduced a simulation-based framework that evaluates the influence of coordinated accounts directly on retweet networks by comparing their observed positions with meaningful baseline strategies.

Our analyses reveal a nuanced picture of the influence exerted by coordinated accounts. In the information cascades associated with the 2019 UK elections, coordinated accounts generally achieve only a small fraction of the theoretically attainable influence. In contrast, the analysis of retweet networks extracted from verified information operations shows that coordinated accounts often occupy structurally advantageous positions, achieving influence comparable to, and in some cases exceeding, that of structurally central nodes identified by the $k$-shell heuristic. Moreover, by reconstructing cascades from the same retweet networks, we show that the two frameworks produce coherent results, suggesting that these findings reflect structural properties of the networks rather than artifacts of the modeling approach. The discrepancy observed for the UK cascades may instead be explained either by the specific characteristics of the analyzed campaign, which may not have been particularly effective, or by limitations in the coordinated account detection procedure used to identify coordinated users.

Our work also has several limitations. First, the reconstruction of information cascades necessarily relies on heuristics, as the true diffusion paths are generally unavailable from publicly accessible social media data. More broadly, detailed cascade data remain scarce, limiting the applicability of exact cascade-based analyses. Second, our empirical evaluation is restricted to Twitter/X, and the conclusions may not directly generalize to other social platforms, where interaction mechanisms and coordination strategies may differ substantially. Finally, our framework is inherently post-hoc: it quantifies the influence of coordinated behavior only after diffusion has taken place and, in its current form, cannot be used to assess the impact of coordinated campaigns in real time. A further limitation concerns the nature of the labels used in the analysis. In the UK election dataset, coordinated accounts are identified through a behavioral detection pipeline and should therefore be interpreted as accounts exhibiting coordinated behavior rather than verified inauthentic behavior. In contrast, the information-operation datasets contain accounts released by the platform as part of state-backed operations. Moreover, retweet networks aggregate interactions over time, whereas the independent cascade model abstracts away from temporal ordering, repeated exposures, platform ranking mechanisms, content heterogeneity, and user-level differences in susceptibility. The simulation results should therefore be interpreted as estimates of structural diffusion potential rather than as reconstructions of the actual temporal spreading process. 

Despite these limitations, the proposed frameworks provide practical tools for the retrospective evaluation of coordinated behavior. In particular, the polynomial-time optimal algorithm and the network-based simulation framework could be directly applied by social media platforms or other organizations with access to complete diffusion data, enabling a principled assessment of the effectiveness of coordinated campaigns. Such evaluations could complement existing detection methods by quantifying not only whether coordinated behavior exists, but also its potential impact on information diffusion, thereby supporting more informed mitigation strategies against influence operations.

% \section{Acknowledgments}


% \begin{acks}

% \end{acks}

\bibliographystyle{ACM-Reference-Format}
\bibliography{bibliography}

@article{DiMarco2025,
  title = {Post-hoc Evaluation of Nodes Influence in Information Cascades: The Case of Coordinated Accounts},
  volume = {19},
  ISSN = {1559-114X},
  url = {http://dx.doi.org/10.1145/3700644},
  DOI = {10.1145/3700644},
  number = {2},
  journal = {ACM Transactions on the Web},
  publisher = {Association for Computing Machinery (ACM)},
  author = {Di Marco,  Niccolò and Brunetti,  Sara and Cinelli,  Matteo and Quattrociocchi,  Walter},
  year = {2025},
  month = may,
  pages = {1–19}
}

@inproceedings{hristakieva2022spread,
  title={The spread of propaganda by coordinated communities on social media},
  author={Hristakieva, Kristina and Cresci, Stefano and Da San Martino, Giovanni and Conti, Mauro and Nakov, Preslav},
  booktitle={Proceedings of the 14th ACM Web Science Conference 2022},
  pages={191--201},
  year={2022}
}

@inproceedings{pacheco2021uncovering,
  title={Uncovering coordinated networks on social media: methods and case studies},
  author={Pacheco, Diogo and Hui, Pik-Mai and Torres-Lugo, Christopher and Truong, Bao Tran and Flammini, Alessandro and Menczer, Filippo},
  booktitle={Proceedings of the international AAAI conference on web and social media},
  volume={15},
  pages={455--466},
  year={2021}
}

@article{luceri2023unmasking,
  title={Unmasking the Web of Deceit: Uncovering Coordinated Activity to Expose Information Operations on Twitter},
  author={Luceri, Luca and Pant{\`e}, Valeria and Burghardt, Keith and Ferrara, Emilio},
  journal={arXiv preprint arXiv:2310.09884},
  year={2023}
}

@inbook{DiMarco2025_Constrained_Maximization_Tree,
  title = {On the Constrained Maximization of Influence Over a Directed Tree Structure},
  ISBN = {9783032095442},
  ISSN = {1611-3349},
  url = {http://dx.doi.org/10.1007/978-3-032-09544-2_17},
  DOI = {10.1007/978-3-032-09544-2_17},
  booktitle = {Discrete Geometry and Mathematical Morphology},
  publisher = {Springer Nature Switzerland},
  author = {Di Marco,  Niccolò and Frosini,  Andrea and Nakano,  Shinichi},
  year = {2025},
  month = nov,
  pages = {237–248}
}

@article{nwala2023language,
  title={A language framework for modeling social media account behavior},
  author={Nwala, Alexander C and Flammini, Alessandro and Menczer, Filippo},
  journal={EPJ Data Science},
  volume={12},
  number={1},
  pages={33},
  year={2023},
  publisher={Springer Berlin Heidelberg}
}

@misc{dataset_menczer,
  doi = {10.5281/ZENODO.14189193},
  url = {https://zenodo.org/doi/10.5281/zenodo.14189193},
  author = {Se\c{c}kin,  \"{O}zg\"{u}r Can and Pote,  Manita and Nwala,  Alexander and Yin,  Lake and Luceri,  Luca and flammini,  alessandro and Menczer,  Filippo},
  title = {Labeled Datasets for Research on Information Operations},
  publisher = {Zenodo},
  year = {2024},
  copyright = {Creative Commons Attribution Non Commercial No Derivatives 4.0 International}
}

@inproceedings{Kempe2003,
  series = {KDD03},
  title = {Maximizing the spread of influence through a social network},
  url = {http://dx.doi.org/10.1145/956750.956769},
  DOI = {10.1145/956750.956769},
  booktitle = {Proceedings of the ninth ACM SIGKDD international conference on Knowledge discovery and data mining},
  publisher = {ACM},
  author = {Kempe,  David and Kleinberg,  Jon and Tardos,  Éva},
  year = {2003},
  month = aug,
  collection = {KDD03}
}

@article{Serrano2009,
  title = {Extracting the multiscale backbone of complex weighted networks},
  volume = {106},
  ISSN = {1091-6490},
  url = {http://dx.doi.org/10.1073/pnas.0808904106},
  DOI = {10.1073/pnas.0808904106},
  number = {16},
  journal = {Proceedings of the National Academy of Sciences},
  publisher = {Proceedings of the National Academy of Sciences},
  author = {Serrano,  M. Ángeles and Boguñá,  Marián and Vespignani,  Alessandro},
  year = {2009},
  month = apr,
  pages = {6483–6488}
}

@inbook{Kempe2005,
  title = {Influential Nodes in a Diffusion Model for Social Networks},
  ISBN = {9783540316916},
  ISSN = {1611-3349},
  url = {http://dx.doi.org/10.1007/11523468_91},
  DOI = {10.1007/11523468_91},
  booktitle = {Automata,  Languages and Programming},
  publisher = {Springer Berlin Heidelberg},
  author = {Kempe,  David and Kleinberg,  Jon and Tardos,  Éva},
  year = {2005},
  pages = {1127–1138}
}

@inproceedings{Mossel2007,
  series = {STOC07},
  title = {On the submodularity of influence in social networks},
  url = {http://dx.doi.org/10.1145/1250790.1250811},
  DOI = {10.1145/1250790.1250811},
  booktitle = {Proceedings of the thirty-ninth annual ACM symposium on Theory of computing},
  publisher = {ACM},
  author = {Mossel,  Elchanan and Roch,  Sebastien},
  year = {2007},
  month = jun,
  pages = {128–134},
  collection = {STOC07}
}

@article{Kitsak2010,
  title = {Identification of influential spreaders in complex networks},
  volume = {6},
  ISSN = {1745-2481},
  url = {http://dx.doi.org/10.1038/nphys1746},
  DOI = {10.1038/nphys1746},
  number = {11},
  journal = {Nature Physics},
  publisher = {Springer Science and Business Media LLC},
  author = {Kitsak,  Maksim and Gallos,  Lazaros K. and Havlin,  Shlomo and Liljeros,  Fredrik and Muchnik,  Lev and Stanley,  H. Eugene and Makse,  Hernán A.},
  year = {2010},
  month = aug,
  pages = {888–893}
}

@article{Seidman1983,
  title = {Network structure and minimum degree},
  volume = {5},
  ISSN = {0378-8733},
  url = {http://dx.doi.org/10.1016/0378-8733(83)90028-X},
  DOI = {10.1016/0378-8733(83)90028-x},
  number = {3},
  journal = {Social Networks},
  publisher = {Elsevier BV},
  author = {Seidman,  Stephen B.},
  year = {1983},
  month = sep,
  pages = {269–287}
}

@article{Carmi2007,
  title = {A model of Internet topology usingk-shell decomposition},
  volume = {104},
  ISSN = {1091-6490},
  url = {http://dx.doi.org/10.1073/pnas.0701175104},
  DOI = {10.1073/pnas.0701175104},
  number = {27},
  journal = {Proceedings of the National Academy of Sciences},
  publisher = {Proceedings of the National Academy of Sciences},
  author = {Carmi,  Shai and Havlin,  Shlomo and Kirkpatrick,  Scott and Shavitt,  Yuval and Shir,  Eran},
  year = {2007},
  month = jul,
  pages = {11150–11154}
}

@article{cinelli2022coordinated,
  title={Coordinated inauthentic behavior and information spreading on twitter},
  author={Cinelli, Matteo and Cresci, Stefano and Quattrociocchi, Walter and Tesconi, Maurizio and Zola, Paola},
  journal={Decision Support Systems},
  volume={160},
  pages={113819},
  year={2022},
  publisher={Elsevier}
}

@article{Tardelli2024temporal,
  title = {Temporal dynamics of coordinated online behavior: Stability,  archetypes,  and influence},
  volume = {121},
  ISSN = {1091-6490},
  url = {http://dx.doi.org/10.1073/pnas.2307038121},
  DOI = {10.1073/pnas.2307038121},
  number = {20},
  journal = {Proceedings of the National Academy of Sciences},
  publisher = {Proceedings of the National Academy of Sciences},
  author = {Tardelli,  Serena and Nizzoli,  Leonardo and Tesconi,  Maurizio and Conti,  Mauro and Nakov,  Preslav and Da San Martino,  Giovanni and Cresci,  Stefano},
  year = {2024},
  month = may 
}

@article{Cinelli2022,
  title = {Coordinated inauthentic behavior and information spreading on Twitter},
  volume = {160},
  ISSN = {0167-9236},
  url = {http://dx.doi.org/10.1016/j.dss.2022.113819},
  DOI = {10.1016/j.dss.2022.113819},
  journal = {Decision Support Systems},
  publisher = {Elsevier BV},
  author = {Cinelli,  Matteo and Cresci,  Stefano and Quattrociocchi,  Walter and Tesconi,  Maurizio and Zola,  Paola},
  year = {2022},
  month = sep,
  pages = {113819}
}

@article{weber2021amplifying,
	author = {Weber, Derek and Neumann, Frank},
	journal = {Social Network Analysis and Mining},
	number = {1},
	year = {2021},
	pages = {},
	publisher = {{Springer Science and Business Media LLC}},
	title = {Amplifying influence through coordinated behaviour in social networks},
	volume = {11},
}

@article{ferrara2016rise,
  title={The rise of social bots},
  author={Ferrara, Emilio and Varol, Onur and Davis, Clayton and Menczer, Filippo and Flammini, Alessandro},
  journal={Communications of the ACM},
  volume={59},
  number={7},
  pages={96--104},
  year={2016},
  publisher={ACM New York, NY, USA}
}

@article{Cresci2020decade,
  title = {A decade of social bot detection},
  volume = {63},
  ISSN = {1557-7317},
  url = {http://dx.doi.org/10.1145/3409116},
  DOI = {10.1145/3409116},
  number = {10},
  journal = {Communications of the ACM},
  publisher = {Association for Computing Machinery (ACM)},
  author = {Cresci,  Stefano},
  year = {2020},
  month = sep,
  pages = {72–83}
}

@inproceedings{Cheng2017anyone,
  series = {CSCW ’17},
  title = {Anyone Can Become a Troll: Causes of Trolling Behavior in Online Discussions},
  url = {http://dx.doi.org/10.1145/2998181.2998213},
  DOI = {10.1145/2998181.2998213},
  booktitle = {Proceedings of the 2017 ACM Conference on Computer Supported Cooperative Work and Social Computing},
  publisher = {ACM},
  author = {Cheng,  Justin and Bernstein,  Michael and Danescu-Niculescu-Mizil,  Cristian and Leskovec,  Jure},
  year = {2017},
  month = feb,
  pages = {1217–1230},
  collection = {CSCW ’17}
}

@article{Nizzoli2021,
  title = {Coordinated Behavior on Social Media in 2019 UK General Election},
  volume = {15},
  ISSN = {2162-3449},
  url = {http://dx.doi.org/10.1609/icwsm.v15i1.18074},
  DOI = {10.1609/icwsm.v15i1.18074},
  journal = {Proceedings of the International AAAI Conference on Web and Social Media},
  publisher = {Association for the Advancement of Artificial Intelligence (AAAI)},
  author = {Nizzoli,  Leonardo and Tardelli,  Serena and Avvenuti,  Marco and Cresci,  Stefano and Tesconi,  Maurizio},
  year = {2021},
  month = may,
  pages = {443–454}
}

@inbook{Shu2020,
  title = {Mining Disinformation and Fake News: Concepts,  Methods,  and Recent Advancements},
  ISBN = {9783030426996},
  ISSN = {2190-5436},
  url = {http://dx.doi.org/10.1007/978-3-030-42699-6_1},
  DOI = {10.1007/978-3-030-42699-6_1},
  booktitle = {Disinformation,  Misinformation,  and Fake News in Social Media},
  publisher = {Springer International Publishing},
  author = {Shu,  Kai and Wang,  Suhang and Lee,  Dongwon and Liu,  Huan},
  year = {2020},
  pages = {1–19}
}

@article{Starbird2019disinformation,
  title = {Disinformation as Collaborative Work: Surfacing the Participatory Nature of Strategic Information Operations},
  volume = {3},
  ISSN = {2573-0142},
  url = {http://dx.doi.org/10.1145/3359229},
  DOI = {10.1145/3359229},
  number = {CSCW},
  journal = {Proceedings of the ACM on Human-Computer Interaction},
  publisher = {Association for Computing Machinery (ACM)},
  author = {Starbird,  Kate and Arif,  Ahmer and Wilson,  Tom},
  year = {2019},
  month = nov,
  pages = {1–26}
}

@article{stella2018bots,
  title={Bots increase exposure to negative and inflammatory content in online social systems},
  author={Stella, Massimo and Ferrara, Emilio and De Domenico, Manlio},
  journal={Proceedings of the National Academy of Sciences},
  volume={115},
  number={49},
  pages={12435--12440},
  year={2018},
  publisher={National Academy of Sciences}
}

@inproceedings{luceri2019red,
  title={Red bots do it better: Comparative analysis of social bot partisan behavior},
  author={Luceri, Luca and Deb, Ashok and Badawy, Adam and Ferrara, Emilio},
  booktitle={Companion proceedings of the 2019 world wide web conference},
  pages={1007--1012},
  year={2019}
}

@article{keller2020political,
	author = {Keller, Franziska B. and Schoch, David and Stier, Sebastian and Yang, JungHwan},
	journal = {Political Communication},
	number = {2},
	year = {2020},
	pages = {256--280},
	publisher = {Informa UK Limited},
	title = {Political {Astroturfing} on {Twitter}: How to {Coordinate} a {Disinformation} {Campaign}},
	volume = {37},
}

@article{schoch2022coordination,
	author = {Schoch, David and Keller, Franziska B. and Stier, Sebastian and Yang, JungHwan},
	journal = {Scientific Reports},
	number = {1},
	year = {2022},
	pages = {},
	publisher = {{Springer Science and Business Media LLC}},
	title = {Coordination patterns reveal online political astroturfing across the world},
	volume = {12},
}

@article{Ratkiewicz2021detecting,
  title = {Detecting and Tracking Political Abuse in Social Media},
  volume = {5},
  ISSN = {2162-3449},
  url = {http://dx.doi.org/10.1609/icwsm.v5i1.14127},
  DOI = {10.1609/icwsm.v5i1.14127},
  number = {1},
  journal = {Proceedings of the International AAAI Conference on Web and Social Media},
  publisher = {Association for the Advancement of Artificial Intelligence (AAAI)},
  author = {Ratkiewicz,  Jacob and Conover,  Michael and Meiss,  Mark and Goncalves,  Bruno and Flammini,  Alessandro and Menczer,  Filippo},
  year = {2021},
  month = aug,
  pages = {297–304}
}

@inproceedings{ratkiewicz2011truthy,
  title={Truthy: mapping the spread of astroturf in microblog streams},
  author={Ratkiewicz, Jacob and Conover, Michael and Meiss, Mark and Gon{\c{c}}alves, Bruno and Patil, Snehal and Flammini, Alessandro and Menczer, Filippo},
  booktitle={Proceedings of the 20th international conference companion on World wide web},
  pages={249--252},
  year={2011}
}

@inproceedings{badawy2018analyzing,
author = {Badawy, Adam and Ferrara, Emilio and Lerman, Kristina},
title = {Analyzing the digital traces of political manipulation: the 2016 russian interference Twitter campaign},
year = {2020},
isbn = {9781538660515},
publisher = {IEEE Press},
booktitle = {Proceedings of the 2018 IEEE/ACM International Conference on Advances in Social Networks Analysis and Mining},
pages = {258–265},
numpages = {8},
location = {Barcelona, Spain},
series = {ASONAM '18}
}

@article{Cinelli2019information,
 ISSN = {14453312, 14453347},
 URL = {https://www.jstor.org/stable/26894683},
 author = {M Cinelli and M Conti and L Finos and F Grisolia and P Kralj Novak and A Peruzzi and M Tesconi and F Zollo and W Quattrociocchi},
 journal = {Journal of Information Warfare},
 number = {3},
 pages = {83--98},
 publisher = {Peregrine Technical Solutions},
 title = {(Mis)Information Operations: An Integrated Perspective},
 urldate = {2025-06-19},
 volume = {18},
 year = {2019}
}

@article{bennett2012logic,
  title={The logic of connective action: Digital media and the personalization of contentious politics},
  author={Bennett, W Lance and Segerberg, Alexandra},
  journal={Information, communication \& society},
  volume={15},
  number={5},
  pages={739--768},
  year={2012},
  publisher={Taylor \& Francis}
}

@book{weller2013twitter,
  author = {Katrin Welle and Axel Bruns and Jean Burgess and Merja Mahrt and Cornelius Puschmann},
  title = {Twitter and Society},
  ISBN = {9781433121708},
  url = {http://dx.doi.org/10.3726/978-1-4539-1170-9},
  DOI = {10.3726/978-1-4539-1170-9},
  publisher = {Peter Lang US},
  year = {2014},
  month = jan 
}

@article{Centola2018,
  title = {Experimental evidence for tipping points in social convention},
  volume = {360},
  ISSN = {1095-9203},
  url = {http://dx.doi.org/10.1126/science.aas8827},
  DOI = {10.1126/science.aas8827},
  number = {6393},
  journal = {Science},
  publisher = {American Association for the Advancement of Science (AAAS)},
  author = {Centola,  Damon and Becker,  Joshua and Brackbill,  Devon and Baronchelli,  Andrea},
  year = {2018},
  month = jun,
  pages = {1116–1119}
}

@inproceedings{castillo2014characterizing,
  series = {CSCW’14},
  title = {Characterizing the life cycle of online news stories using social media reactions},
  url = {http://dx.doi.org/10.1145/2531602.2531623},
  DOI = {10.1145/2531602.2531623},
  booktitle = {Proceedings of the 17th ACM conference on Computer supported cooperative work \&amp; social computing},
  publisher = {ACM},
  author = {Castillo,  Carlos and El-Haddad,  Mohammed and Pfeffer,  Jurgen and Stempeck,  Matt},
  year = {2014},
  month = feb,
  pages = {211–223},
  collection = {CSCW’14}
}

@article{Ashery2025,
  title = {Emergent social conventions and collective bias in LLM populations},
  volume = {11},
  ISSN = {2375-2548},
  url = {http://dx.doi.org/10.1126/sciadv.adu9368},
  DOI = {10.1126/sciadv.adu9368},
  number = {20},
  journal = {Science Advances},
  publisher = {American Association for the Advancement of Science (AAAS)},
  author = {Ashery,  Ariel Flint and Aiello,  Luca Maria and Baronchelli,  Andrea},
  year = {2025},
  month = may 
}

@article{Flint2025,
  doi = {10.48550/ARXIV.2510.22422},
  url = {https://arxiv.org/abs/2510.22422},
  author = {Flint,  Ariel and Aiello,  Luca Maria and Pastor-Satorras,  Romualdo and Baronchelli,  Andrea},
  title = {Group size effects and collective misalignment in LLM multi-agent systems},
  publisher = {arXiv},
  year = {2025},
  copyright = {arXiv.org perpetual,  non-exclusive license}
}

@article{Schroeder2025,
  doi = {10.48550/ARXIV.2506.06299},
  url = {https://arxiv.org/abs/2506.06299},
  author = {Schroeder,  Daniel Thilo and Cha,  Meeyoung and Baronchelli,  Andrea and Bostrom,  Nick and Christakis,  Nicholas A. and Garcia,  David and Goldenberg,  Amit and Kyrychenko,  Yara and Leyton-Brown,  Kevin and Lutz,  Nina and Marcus,  Gary and Menczer,  Filippo and Pennycook,  Gordon and Rand,  David G. and Ressa,  Maria and Schweitzer,  Frank and Summerfield,  Christopher and Tang,  Audrey and Van Bavel,  Jay J. and van der Linden,  Sander and Song,  Dawn and Kunst,  Jonas R.},
  title = {How Malicious AI Swarms Can Threaten Democracy: The Fusion of Agentic AI and LLMs Marks a New Frontier in Information Warfare},
  publisher = {arXiv},
  year = {2025},
  copyright = {arXiv.org perpetual,  non-exclusive license}
}

@article{Orland2025,
  doi = {10.48550/ARXIV.2510.25003},
  url = {https://arxiv.org/abs/2510.25003},
  author = {Orlando,  Gian Marco and Ye,  Jinyi and La Gatta,  Valerio and Saeedi,  Mahdi and Moscato,  Vincenzo and Ferrara,  Emilio and Luceri,  Luca},
  title = {Emergent Coordinated Behaviors in Networked LLM Agents: Modeling the Strategic Dynamics of Information Operations},
  publisher = {arXiv},
  year = {2025},
  copyright = {Creative Commons Attribution Share Alike 4.0 International}
}

@inproceedings{alothali2018detecting,
  title = {Detecting Social Bots on Twitter: A Literature Review},
  url = {http://dx.doi.org/10.1109/INNOVATIONS.2018.8605995},
  DOI = {10.1109/innovations.2018.8605995},
  booktitle = {2018 International Conference on Innovations in Information Technology (IIT)},
  publisher = {IEEE},
  author = {Alothali,  Eiman and Zaki,  Nazar and Mohamed,  Elfadil A. and Alashwal,  Hany},
  year = {2018},
  month = nov,
  pages = {175–180}
}

@inproceedings{Pant2025beyond,
  series = {WWW ’25},
  title = {Beyond Interaction Patterns: Assessing Claims of Coordinated Inter-State Information Operations on Twitter/X},
  url = {http://dx.doi.org/10.1145/3701716.3715575},
  DOI = {10.1145/3701716.3715575},
  booktitle = {Companion Proceedings of the ACM on Web Conference 2025},
  publisher = {ACM},
  author = {Pantè,  Valeria and Axelrod,  David and Flammini,  Alessandro and Menczer,  Filippo and Ferrara,  Emilio and Luceri,  Luca},
  year = {2025},
  month = may,
  pages = {1234–1238},
  collection = {WWW ’25}
}

@article{Minici2025iohunter,
  title = {IOHunter: Graph Foundation Model to Uncover Online Information Operations},
  volume = {39},
  ISSN = {2159-5399},
  url = {http://dx.doi.org/10.1609/aaai.v39i27.35046},
  DOI = {10.1609/aaai.v39i27.35046},
  number = {27},
  journal = {Proceedings of the AAAI Conference on Artificial Intelligence},
  publisher = {Association for the Advancement of Artificial Intelligence (AAAI)},
  author = {Minici,  Marco and Luceri,  Luca and Fabbri,  Francesco and Ferrara,  Emilio},
  year = {2025},
  month = apr,
  pages = {28258–28266}
}

@inproceedings{Cinus2025exposing,
  series = {WWW ’25},
  title = {Exposing Cross-Platform Coordinated Inauthentic Activity in the Run-Up to the 2024 U.S. Election},
  url = {http://dx.doi.org/10.1145/3696410.3714698},
  DOI = {10.1145/3696410.3714698},
  booktitle = {Proceedings of the ACM on Web Conference 2025},
  publisher = {ACM},
  author = {Cinus,  Federico and Minici,  Marco and Luceri,  Luca and Ferrara,  Emilio},
  year = {2025},
  month = apr,
  pages = {541–559},
  collection = {WWW ’25}
}

@article{Cinelli2021dynamics,
  title = {Dynamics of online hate and misinformation},
  volume = {11},
  ISSN = {2045-2322},
  url = {http://dx.doi.org/10.1038/s41598-021-01487-w},
  DOI = {10.1038/s41598-021-01487-w},
  number = {1},
  journal = {Scientific Reports},
  publisher = {Springer Science and Business Media LLC},
  author = {Cinelli,  Matteo and Pelicon,  Andraž and Mozetič,  Igor and Quattrociocchi,  Walter and Novak,  Petra Kralj and Zollo,  Fabiana},
  year = {2021},
  month = nov 
}

@article{Loru2024,
  title = {The influence of coordinated behavior on toxicity},
  volume = {43–44},
  ISSN = {2468-6964},
  url = {http://dx.doi.org/10.1016/j.osnem.2024.100289},
  DOI = {10.1016/j.osnem.2024.100289},
  journal = {Online Social Networks and Media},
  publisher = {Elsevier BV},
  author = {Loru,  Edoardo and Cinelli,  Matteo and Tesconi,  Maurizio and Quattrociocchi,  Walter},
  year = {2024},
  month = nov,
  pages = {100289}
}

@article{Cresci2019capability,
  title = {On the capability of evolved spambots to evade detection via genetic engineering},
  volume = {9},
  ISSN = {2468-6964},
  url = {http://dx.doi.org/10.1016/j.osnem.2018.10.005},
  DOI = {10.1016/j.osnem.2018.10.005},
  journal = {Online Social Networks and Media},
  publisher = {Elsevier BV},
  author = {Cresci,  Stefano and Petrocchi,  Marinella and Spognardi,  Angelo and Tognazzi,  Stefano},
  year = {2019},
  month = jan,
  pages = {1–16}
}

@misc{CIB_Survey,
  doi = {10.48550/ARXIV.2408.01257},
  url = {https://arxiv.org/abs/2408.01257},
  author = {Mannocci,  Lorenzo and Mazza,  Michele and Monreale,  Anna and Tesconi,  Maurizio and Cresci,  Stefano},
  title = {Detection and Characterization of Coordinated Online Behavior: A Survey},
  publisher = {arXiv},
  year = {2024},
  copyright = {arXiv.org perpetual,  non-exclusive license}
}

@article{Loru2026,
  title = {A Compression-Based Approach to Detecting Automated and Coordinated Behavior on Social Media},
  volume = {20},
  ISSN = {1556-472X},
  url = {http://dx.doi.org/10.1145/3778356},
  DOI = {10.1145/3778356},
  number = {2},
  journal = {ACM Transactions on Knowledge Discovery from Data},
  publisher = {Association for Computing Machinery (ACM)},
  author = {Loru,  Edoardo and Di Marco,  Niccolò and Cinelli,  Matteo and Quattrociocchi,  Walter},
  year = {2026},
  month = Jan,
  pages = {1–25}
}

@article{woolley2020bots,
  title={Bots and computational propaganda: Automation for communication and control},
  author={Woolley, Samuel C},
  journal={Social media and democracy: The state of the field, prospects for reform},
  pages={89--110},
  year={2020},
  publisher={Cambridge University Press Cambridge}
}

@article{eady2023exposure,
  title={Exposure to the Russian Internet Research Agency foreign influence campaign on Twitter in the 2016 US election and its relationship to attitudes and voting behavior},
  author={Eady, Gregory and Paskhalis, Tom and Zilinsky, Jan and Bonneau, Richard and Nagler, Jonathan and Tucker, Joshua A},
  journal={Nature communications},
  volume={14},
  number={1},
  pages={62},
  year={2023},
  publisher={Nature Publishing Group UK London}
}

@article{guardian_russia_2024,
  author       = {Borrger, Julian and Singh, Maanvi and Pengelly, Martin},
  title        = {US accuses Russia of elaborate online campaigns to influence 2024 voters},
  journal      = {The Guardian},
  year         = {2024},
  month        = sep,
  day          = {4},
  url          = {https://www.theguardian.com/us-news/article/2024/sep/04/russia-accused-of-trying-to-influence-us-voters-through-online-campaign},
  note         = {Accessed: June 22, 2026}
}

@article{serafino2024suspended,
  title={Suspended accounts align with the Internet Research Agency misinformation campaign to influence the 2016 US election},
  author={Serafino, Matteo and Zhou, Zhenkun and Andrade Jr, Jos{\'e} S and Bovet, Alexandre and Makse, Hern{\'a}n},
  journal={EPJ Data Science},
  volume={13},
  number={1},
  pages={29},
  year={2024},
  publisher={Springer}
}

@article{wagner2023independence,
  title={Independence by permission},
  author={Wagner, Michael W},
  journal={Science},
  volume={381},
  number={6656},
  pages={388--391},
  year={2023},
  publisher={American Association for the Advancement of Science}
}

@misc{gleicher2018coordinated,
  author       = {Gleicher, Nathaniel},
  title        = {Coordinated Inauthentic Behavior Explained},
  year         = {2018},
  month        = dec,
  day          = {6},
  howpublished = {Meta Newsroom},
  url          = {https://about.fb.com/news/2018/12/inside-feed-coordinated-inauthentic-behavior/},
  note         = {Updated August 29, 2022}
}

@article{Kupferschmidt2026,
  title = {The misinformation accelerator},
  volume = {392},
  ISSN = {1095-9203},
  url = {http://dx.doi.org/10.1126/science.aei4485},
  DOI = {10.1126/science.aei4485},
  number = {6797},
  journal = {Science},
  publisher = {American Association for the Advancement of Science (AAAS)},
  author = {Kupferschmidt,  Kai},
  year = {2026},
  month = Apr,
  pages = {456–459}
}

\clearpage
\newpage
\section{Appendices}

% If your work needs an appendix, add it before the
% ``\verb|\end{document}|'' command at the conclusion of your source
% document.

% Start the appendix with the ``\verb|appendix|'' command:
% \begin{verbatim}
%   \appendix
% \end{verbatim}
% and note that in the appendix, sections are lettered, not
% numbered.
\appendix

\section{Compare the performances of greedy and optimal algorithms}\label{sec:comparison}
In the main paper, we introduce an optimal algorithm capable of computing the maximum influence over a directed tree with a fixed number of $1$-nodes, in contrast to the greedy approach proposed in \cite{DiMarco2025}. In this section, we compare the performance of the two methods on synthetic data in order to evaluate their efficiency.

To this end, we generate four trees with sizes of 50, 100, 500, and 1000 nodes. For each tree, we consider four values of $k$, corresponding to $20\%$, $40\%$, $60\%$, and $80\%$ of the total number of nodes. We then run both the greedy and the optimal algorithms 10 times for each configuration, recording their execution times.

Table \ref{tab:greedy_vs_optimal_results} presents the results obtained on the synthetic trees. In particular, each cell reports the relative gap between the optimal and greedy solutions, defined as 

\[
\frac{I_{\mathrm{opt}}(c,k)-I_{\mathrm{greedy}}(c,k)}
{I_{\mathrm{opt}}(c,k)}\cdot 100.
\]

\begin{table}[!ht]
\centering
\caption{Relative gap (\%) of the greedy algorithm with respect to the optimal solution for different values of $N$ and $k$.} 
\label{tab:greedy_vs_optimal_results}
\begin{tabular}{ccccc}
  \hline
$N$ & $k = 0.2N$ & $k = 0.4N$ & $k = 0.6N$ & $k = 0.8N$ \\ 
  \hline
50 & 0.00 & 0.00 & 0.00 & 0.00 \\ 
  100 & 0.00 & 0.00 & 0.00 & 0.00 \\ 
  500 & 0.00 & 0.00 & 0.00 & 0.00 \\ 
  1000 & 0.78 & 0.51 & 0.00 & 0.00 \\ 
   \hline
\end{tabular}
\end{table}

The results clearly show that the greedy algorithm achieves values very close to the optimal ones, with only a few exceptions where the relative gap remains small. This indicates that the influence obtained is generally comparable to the optimal solution.

Given this, it is natural to evaluate the execution time of the two approaches to understand whether one approach is more preferable than the other. 
Figure \ref{fig:execution_time} presents the execution time of the algorithms. Each panel corresponds to a tree of a given size. The vertical bars indicate the minimum and maximum values, while the boxes represent the interquartile range (IQR), with the points denoting the median execution time.

\begin{figure}[!ht]
    \centering
    \includegraphics[width=0.9\linewidth]{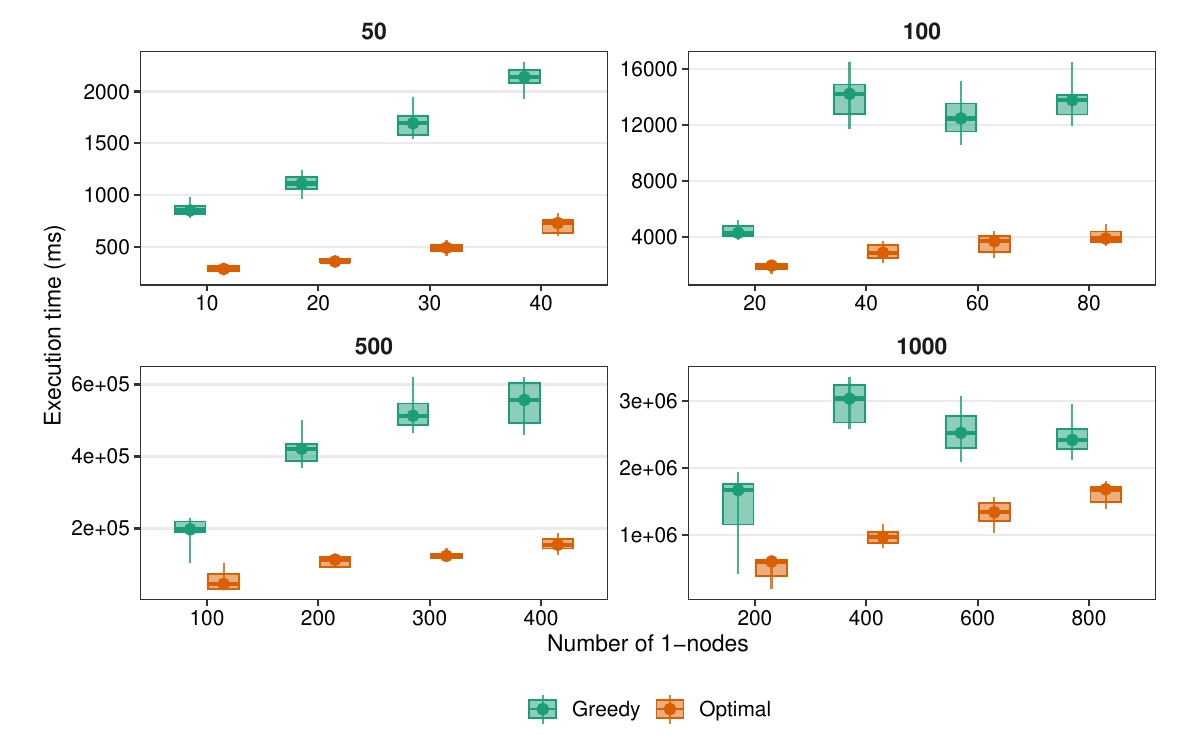}
    \caption{Distribution of execution times for the greedy and optimal approaches. Each panel corresponds to a tree of a given size. The vertical bars indicate the minimum and maximum values, while the boxes represent the interquartile range (IQR), with the points denoting the median execution time.}
    \label{fig:execution_time}
\end{figure}

We observe that the execution times of the greedy algorithm are generally higher than those of the optimal DP approach, with the gap widening as $k$ increases.
Taken together, these results indicate that, while the two methods yield comparable outcomes in terms of influence, the optimal DP approach is preferable in real-world applications due to its lower execution time.

\section{Robustness analysis with different infection probabilities}\label{sec:robustness_retweet_networks}
In the main paper, we evaluate influence on retweet networks using an infection probability of $p = 0.3$ in the independent cascade model. Here, we assess the robustness of our findings by considering a wider range of infection probabilities.
More specifically, for each set of \textit{seed nodes} (selected according to the approaches described in the Methods section), we consider infection probabilities $p = \frac{i}{10}$, for $i = 1, \ldots, 9$, and perform $20$ independent simulations for each value of $p$. Figure~\ref{fig:infections_through_p} reports the mean fraction of infected nodes across the $20$ simulations, together with a smoothed trend obtained using LOESS.

\begin{figure}[!ht]
    \centering
    \includegraphics[width=0.9\linewidth]{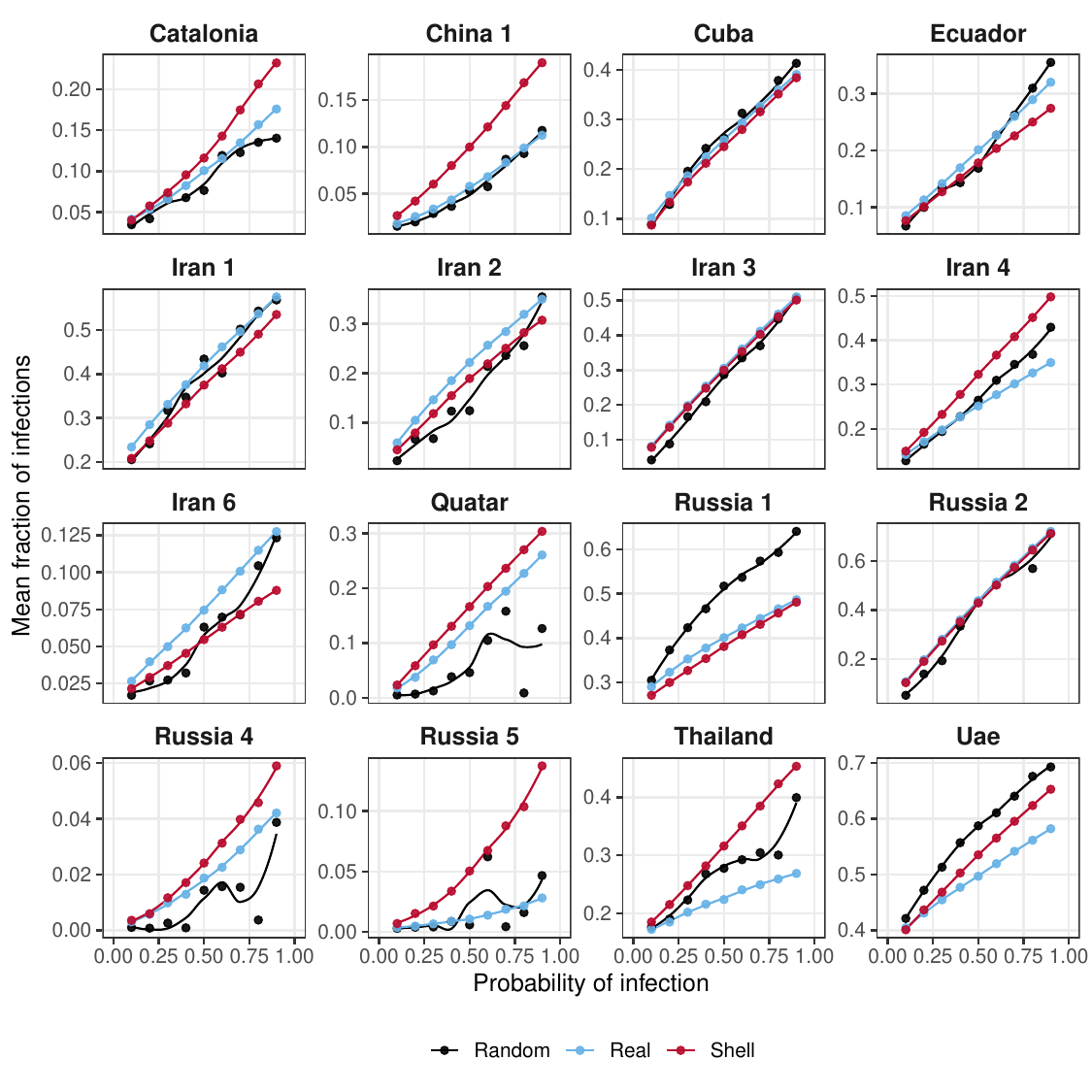}
    \caption{Mean fraction of infected nodes over $20$ independent simulations of the independent cascade model for different infection probabilities $p$ (shown on the $x$-axis). Colors denote the strategy used to select the initial seed nodes.}
    \label{fig:infections_through_p}
\end{figure}

The results indicate that the relative performance of the three strategies remains largely unchanged across the considered range of infection probabilities. As expected, increasing $p$ leads to a larger fraction of infected nodes for all strategies. The stability of these trends supports our decision to present the main analysis using the choice of $p = 0.3$, as the qualitative conclusions remain unchanged under different parameter settings.




\end{document}